\documentclass[11pt]{article}
\usepackage[margin= 1.2 in]{geometry}
\usepackage{inconsolata}
\usepackage{libertine}
\usepackage{amsmath, amssymb}
\usepackage{xargs}
\usepackage{mathtools}
\usepackage{amsthm}
\usepackage{stmaryrd}

\usepackage{graphicx}
\usepackage{hyperref}

\usepackage{framed}
\usepackage[framemethod=TikZ]{mdframed}
\usepackage[noend]{algorithmic}
\usepackage[boxed]{algorithm}

\makeatletter
\let\oldhypertarget\hypertarget
\renewcommand{\hypertarget}[2]{%
	\Hy@raisedlink{\oldhypertarget{#1}{}}
	#2%
	\protected@write\@mainaux{}{%
		\string\expandafter\string\gdef
		\string\csname\string\detokenize{#1}\string\endcsname{#2}%
	}%
}

\newcommand{\mylink}[1]{%
	\hyperlink{#1}{\csname #1\endcsname}%
}
\makeatother

\newcommand{\defproblemu}[3]{
	\vspace{1mm}
	\noindent\fbox{
		\begin{minipage}{0.963\columnwidth}
			#1 \\
			{\bf{Input:}} #2  \\
			{\bf{Question:}} #3
		\end{minipage}
	}
	\vspace{1mm}
}

\usepackage{multirow}
\usepackage{tabularx}

\newcommand{\algcomment}[1]{\colorbox{black!10}{#1}}

\newcommand{\supp}{ \ensuremath{\mathrm{supp}} }

\newtheorem{theorem}{Theorem}[section]
\newtheorem{definition}[theorem]{Definition}
\newtheorem{lemma}[theorem]{Lemma}

\newtheorem{claim}[theorem]{Claim}

\newtheorem{observation}[theorem]{Observation}

\newcommand{\poly}{\mathsf{poly}}

\newcommand{\cA}{\mathcal{A}}
\newcommand{\cB}{\mathcal{B}}

\newcommand{\cF}{\mathcal{F}}
\newcommand{\cH}{\mathcal{H}}
\newcommand{\cP}{\mathcal{P}}
\newcommand{\cS}{\mathcal{S}}

\newcommand{\bS}{\mathbb{S}}

\renewcommand\vec{\mathbf}

\title{Weighted $k$-\textsc{Path} and Other Problems in Almost $O^*(2^k)$ Deterministic Time via Dynamic Representative Sets}
\author{Jesper Nederlof\thanks{Utrecht University, The Netherlands,  \href{j.nederlof@uu.nl}{j.nederlof@uu.nl}. Supported by the project COALESCE that has received funding from the European Research Council (ERC), grant agreement No 853234.}}
\date{}
\begin{document}
\maketitle
\begin{abstract}
We present a data structure that we call a \emph{Dynamic Representative Set}. In its most basic form, it is given two parameters $0< k < n$ and allows us to maintain a representation of a family $\cF$ of subsets of $\{1,\ldots,n\}$. It supports basic update operations (unioning of two families, element convolution) and a query operation that determines for a set $B \subseteq \{1,\ldots,n\}$ whether there is a set $A \in \cF$ of size at most $k-|B|$ such that $A$ and $B$ are disjoint. After $2^{k+O(\sqrt{k}\log^2k)}n \log n$ preprocessing time, all operations use $2^{k+O(\sqrt{k}\log^2k)}\log n$ time.

Our data structure has many algorithmic consequences that improve over previous works.
One application is a deterministic algorithm for the \textsc{Weighted Directed} $k$-\textsc{Path} problem, one of the central problems in parameterized complexity. Our algorithm takes as input an $n$-vertex directed graph $G=(V,E)$ with edge lengths and an integer $k$, and it outputs the minimum edge length of a path on $k$ vertices in $2^{k+O(\sqrt{k}\log^2k)}(n+m)\log n$ time (in the word RAM model where weights fit into a single word). Modulo the lower order term $2^{O(\sqrt{k}\log^2k)}$, this answers a question that has been repeatedly posed as a major open problem in the field.
\end{abstract}
\section{Introduction}
\label{sec:intro}

A central paradigm within parameterized complexity in the last decade or so has been to speed up dynamic programming algorithms by manipulating its tables with \emph{algebraic transformations} or \emph{sparsification}.
Examples of such algebraic transformation methods include the Group Algebra method \cite{DBLP:journals/talg/KoutisW16}, the Subset Convolution method~\cite{DBLP:conf/stoc/BjorklundHKK07}, and the Cut\&Count method~~\cite{DBLP:journals/talg/CyganNPPRW22}.
Examples of the table sparsification method include most notably, Color Coding~\cite{DBLP:journals/jacm/AlonYZ95}, but also Divide\&Color~\cite{doi:10.1137/080716475}, and most relevant for this paper, \emph{representative sets}~\cite{MONIEN1985239,DBLP:journals/iandc/BodlaenderCKN15,DBLP:journals/jacm/FominLPS16}.

While these methods have been extremely well-studied, have been celebrated particularly within the field of Parameterized Complexity,\footnote{For example, the works~\cite{DBLP:journals/jacm/AlonYZ95,DBLP:journals/talg/KratschW14,DBLP:journals/siamcomp/Bjorklund14,DBLP:journals/talg/CyganNPPRW22} all directly relate to either of the two mentioned techniques and have all been awarded a IPEC Nerode prize in separate years.} and found many applications in algorithm design in general, both types of methods still are not so well understood.
In particular, while the algebraic transformations often lead to the fastest known algorithms, these algorithms typically are randomized and inherently do not extend in an efficient way to weighted variants of the problem at hand. In contrast, table sparsification typically \emph{does} lead to deterministic algorithms that also solve weighted variants of the problem, but these algorithms are a bit slower than the ones based on algebraic transformations. A central question therefore is: Can we bridge this gap and obtain the best of both methods?

In this paper we make progress on this central question. We show how to improve the bottleneck in a wide family of algorithms based on representative sets and obtain deterministic algorithms for weighted variants of problems with run times that match the currently fastest (transformation-based) randomized algorithms.

\subsection{The Weighted $k$-\textsc{Path} Problem.}
One of the most well-studied problems in parameterized complexity is the $k$-\textsc{Path} problem.
In its simplest version, one is given an undirected graph $G$ and an integer $k$ and needs to determine whether $G$ has a simple path on $k$ vertices.
There has been a long ``race'' for the fastest algorithm (as a function of $k$) for several natural variants of this problem. See Table~\ref{tab} for an overview.

While the exact result statements of Table~\ref{tab} are of secondary importance, this study has played a key role in the development of central techniques in parameterized complexity such as the Color Coding and Group Algebra method. This also made major impact beyond parameterized complexity.\footnote{For example, color coding found also use in fine-grained complexity (see e.g.~\cite{DBLP:journals/algorithmica/AlonYZ97,DBLP:conf/soda/Bringmann17}) and the algebraic methods were a precursor for a breakthrough result on the \textsc{Hamiltonicity} problem~\cite{DBLP:journals/siamcomp/Bjorklund14}.}

We now state our first algorithmic result formally. We consider the following variant of $k$-\textsc{Path}:

\defproblemu{\textsc{Weighted Directed $k$-Path}}{A directed $n$-vertex graph $G=(V,E)$ with weights $w_{i,j} \in \mathbb{N}$ for $(i,j)\in E$, an integer $k$.}{The minimum weight $\sum_{i=2}^k w_{v_{i-1},v_i}$ of a simple path $(v_1,\ldots,v_k)$ of $G$.}

\begin{theorem}\label{thm:kpath}
	There is a deterministic algorithm for the \textsc{Weighted Directed $k$-path} problem that runs in $2^{k+O(\sqrt{k}\log^2k)}(n+m)\log n$ time on the word RAM model, where each weight fits into a word.
\end{theorem}
The existence of a $O^*(2^k)$ time\footnote{The $O^*()$ notation omits terms polynomial in the input size (where integers are represented in binary).} deterministic algorithm has been \emph{"repeatedly posed as a major open problem in the field"} (quoting~\cite[Page 2]{DBLP:journals/talg/LokshtanovBSZ21}), see e.g.~\cite{DBLP:conf/esa/Zehavi15,DBLP:journals/cacm/KoutisW16, DBLP:journals/jacm/FominLPS16}.
Simultaneously, it was also asked in previous works whether there is an $O^*(2^{k})$ time algorithm that can handle weighted extensions of $k$-\textsc{Path}.
Modulo the lower order $2^{O(\sqrt{k}\log^2k)}$ term, Theorem~\ref{thm:kpath} resolves both questions.

Even for $k=n$, the only known exponential speed-up over $O^*(2^k)$ is an $O^*(1.67^n)$ time algorithm for Undirected Hamiltonicity~\cite{DBLP:journals/siamcomp/Bjorklund14}. Hence it may be hard to improve our algorithm exponentially.

\newcommand{\chec}{\checkmark}
\newcommand{\cross}{$\times$}
\begin{table}[]
	\centering
	\begin{tabular}{|l|l|l|l|l|}
		\hline
		\hspace{0.8em}\parbox[c]{4em}{Run time \vspace{1.7em}} &\rotatebox{90}{\hspace{-0.5em}Determ.}&\rotatebox{90}{\hspace{-0.5em}Weights\ }&\rotatebox{90}{\hspace{-0.5em}Directed}& \hspace{0.1em}\parbox[c]{4em}{Reference\vspace{1.7em}}\\
		\hline
		\hline
		$k!$ & \chec &\chec&\chec&\cite{MONIEN1985239}    \\
		$k!2^k$& \chec &\cross&\cross&\cite{DBLP:conf/wads/Bodlaender89}   \\
		$(2e)^k$& \cross  &\chec&\chec&\cite{DBLP:journals/jacm/AlonYZ95}   \\
		$c^k, c > 8000$ & \chec &\chec&\chec&\cite{DBLP:journals/jacm/AlonYZ95}   \\
		$16^k$ & \chec &\chec&\chec&\cite{DBLP:conf/wg/KneisMRR06}   \\
		$4^k$ & \cross &\chec&\chec&\cite{doi:10.1137/080716475}  \\
		$4^{k+o(k)}$ &\chec&\chec &\chec&\cite{doi:10.1137/080716475}   \\
		$2.83^{k}$ &\cross &\cross&\chec&\cite{DBLP:conf/icalp/Koutis08}  \\
		\hline
	\end{tabular}	
	\ \ 
	\begin{tabular}{|l|l|l|l|l|}
		\hline
		\hspace{0.8em}\parbox[c]{4em}{Run time \vspace{1.7em}} &\rotatebox{90}{\hspace{-0.5em}Determ.}&\rotatebox{90}{\hspace{-0.5em}Weights\ }&\rotatebox{90}{\hspace{-0.5em}Directed}& \hspace{1.4em}\parbox[c]{4em}{Reference\vspace{1.7em}}\\
		\hline
		\hline
		$2^{k}$ & \cross &\cross&\chec&\cite{DBLP:journals/ipl/Williams09},\cite{DBLP:journals/talg/KoutisW16}\\
		$1.66^{k}$ & \cross &\cross&\cross&\cite{DBLP:journals/jcss/BjorklundHKK17}\\		
		$2.851^{k}$ &\chec &\chec&\chec&\cite{DBLP:conf/soda/FominLS14}\\
		$2.62^{k}$ &\chec&\chec &\chec&\cite{DBLP:journals/jacm/FominLPS16},\cite{DBLP:journals/jcss/ShachnaiZ16}\\
		$2.60^{k}$ &\chec &\chec&\chec&\cite{DBLP:conf/esa/Zehavi15} \\		
		$2.56^{k}$ &\chec &\chec&\chec&\cite{DBLP:journals/tcs/Tsur19b}\\
		\hline
		\hline
		\multirow{2}{*}{$2^{k+o(k)}$} &\multirow{2}{*}{\chec}&\multirow{2}{*}{\chec}&\multirow{2}{*}{\chec}&\multirow{2}{*}{\textbf{this work}}\\[2.4ex]
		\hline
	\end{tabular}
	
	\centering
	\caption{The run time of algorithms for $k$-\textsc{Path}, omitting factors polynomial in the input size (where weights are represented in binary). Columns "Determ.", "Weights", "Directed" indicate whether the algorithm is, respectively, deterministic, works for the weighted version (with the stated super-polynomial part of the run time), and works for directed graphs.}
	\label{tab}
\end{table}
\subsection{Representative Sets.}\label{sec:repsets}
Representative sets play a main role in this paper and hence we introduce them in detail now.

Conceptually, the dynamic programming paradigm partitions at any stage a solution into two sets of \emph{partial
	solutions}: A set $\cA$ of partial solutions represented in the current table entries, and a set $\cB$ of (yet unexplored) partial solutions that could potentially complete a partial solution in $\cA$ to a global solution. The definition of $\cA$ and $\cB$ does not depend on the actual input, but only on its parameters (like number of vertices, solution size, or treewidth).

A crucial object of study in the context of representative sets is the \emph{compatibility matrix}, which is a Boolean matrix $\mathbf{M} \in \{0,1\}^{\cA \times \cB}$ indicating whether two (fingerprints of) partial solutions combine into a solution.
Suppose that $\bS$ is a semiring and a dynamic programming table is a vector $\vec{a}\in \bS^\cA$. That is, the table $\vec{a}$ assigns to each fingerprint/partial solution $A \in \cA$ a number $\vec{a}[A]$ that indicates whether there is a partial solution with fingerprint $A$ (if $\bS$ is the Boolean semiring with $\vee$ and $\wedge$ as operations) for decision problems, the number of partial solutions with fingerprint $A$ (if $\bS$ is the ring of integers with integer addition and multiplication) for counting problems, or the minimum weight of a partial solution with fingerprint $A$ (if $\bS$ is the min-sum semiring with $\min$ as addition and integer addition as multiplication) for minimization problems.  

The general idea of representative sets is the following observation: If $\vec{a},\vec{a'} \in \bS^\cA$ are two different dynamic programming tables such that $\vec{a}\cdot\vec{M}=\vec{a'}\cdot\vec{M}$, we call $\vec{a'}$ a \emph{representative set} of $\vec{a}$. Since the solution of the problem at hand equals $\vec{a}\cdot\vec{M}\cdot\vec{b}^\intercal =\vec{a'}\cdot\vec{M}\cdot\vec{b}^\intercal$, for some (unknown) vector $\vec{b}$, we can as well replace the dynamic programming table $\vec{a}$ with its representative set $\vec{a'}$ and retain the solution. The use of this observation is due to that for many matrices $\mathbf{M}$ one can show that any $\vec{a}$ has a representative set $\vec{a'}$ with small support (i.e., for minimization problems, only few values of $\vec{a'}$ are finite). For example, if $\bS$ is the min-plus semiring, a vector $\vec{a}$ with support $\supp(\vec{a})$ is given and $\mathbf{M}$ has rank $r$ over a finite field, one can find a representative set $\vec{a'}$ with $|\supp(\vec{a'})| \leq r$ in time $|\supp(\vec{a})|\cdot r^{\omega-1}$, where $\omega$ is the matrix multiplication constant~\cite{DBLP:journals/iandc/BodlaenderCKN15}.
When the compatibility matrix $\vec{M}$ commutes with the matrices that compute new table entries from old table entries, one can use this observation to compute the outcome of the dynamic programming while only working with representative sets of small support and have faster algorithms.

As mentioned before, unfortunately this approach often does not lead to the fastest known algorithms. The computation of small representative sets is the only bottleneck here, and it is an important and beautiful open question to remove this bottleneck:
\newcounter{oq}
\newenvironment{openq}[1][]{\refstepcounter{oq}\par\medskip\mdfsetup{%
		nobreak=true,
		middlelinecolor=gray,
		middlelinewidth=1pt,
		backgroundcolor=gray!10,
		innertopmargin=5pt,
		roundcorner=5pt}
	\begin{mdframed}\textbf{Question~\theoq:}}{\end{mdframed}\medskip}
\begin{openq}\label{oqa}
	For which $\mathbf{M}$ are there deterministic $\tilde{O}(|\supp(\vec{a})|+r)$ time algorithms to compute a representative set of support size at most $r$ of $\vec{a}$?
\end{openq}
See also~\cite{DBLP:conf/birthday/Nederlof20}.
We stress that this is an important question. For example:
If Question~\ref{oqa} is positively resolved with the compatibility matrix $\mathbf{M}$ indicating whether the union of

\begin{itemize}
	\item two forests of a graph is a spanning tree, then the Cut\&Count method~\cite{DBLP:journals/talg/CyganNPPRW22} can be derandomized and extended to weighted problems without overhead in the runtime,
	\item two independent sets of a matroid form another independent set, then many kernelization algorithms based on~\cite{DBLP:journals/jacm/KratschW20} can be derandomized,
	\item two perfect matchings forms a Hamiltonian cycle, then the approach from~\cite{DBLP:conf/stoc/Nederlof20} is a promising direction towards a $O^*(1.9999^n)$ time algorithm for the \textsc{Traveling Salesperson Problem}.
\end{itemize} 

\subsection{Representative sets for $k$-\textsc{Path}.}\label{subsec:kpath}
All known algorithms for \textsc{Weighted Directed $k$-Path} that run in $O^*(c^k)$ time, for some $c<3$, rely on the representative set method mentioned above.
We will now make this application to \textsc{Weighted Directed $k$-Path} more explicit. A path can, at any vertex $v$ contained in it, be naturally be broken down into two disjoint subpaths ending/starting at $v$, and the concatenation of any path ending at $v$ and starting at $v$ gives another path assuming the paths are disjoint. Hence the compatibility matrix $\vec{M}=\vec{D}_{n,k}$ (defined below) will indicate disjointness.
While this application of representative sets goes back all the way to~\cite{MONIEN1985239}, we present it in a new linear algebraic manner that serves as a gentle introduction to later parts of this paper.

Fix the semiring $\mathbb{S}=(S, +, \cdot,\bar{0},\bar{1})$ to be the min-sum semiring $(\mathbb{N} \cup \infty,\min,+,\infty,0)$ (we will stick to the general notation for convenience).
Let $G=([n],E)$ be a graph and let $w_{i,j} \in S$ be the weight of the edge from vertex $i$ to vertex $j$.
Let $e \in [n]$ and let $k$ be a fixed integer. We define the following ``Element Convolution'' and ``Disjointness'' matrices that are indexed by subsets of $U$:
\[
	\vec{C}_{n,e}[A,B] := \llbracket e \notin A \text{ and } A \cup \{e\} = B \rrbracket \qquad \text{and}\qquad
\vec{D}_{n,k}[A,B] :=	\llbracket A \cap B = \emptyset \text{ and } |A \cup B|\leq k \rrbracket.
\]
Here the Iverson bracket notation $\llbracket b \rrbracket$, for a boolean $b$, indicates $\bar{0}$ if $b$ is false, and $\bar{1}$ if $b$ is true.
A crucial property which we show in Lemma~\ref{lem:comm} is that these matrices commute in the sense that $\vec{C}_{n,e} \cdot  \vec{D} = \vec{D}\cdot  \vec{C}^\intercal_{n,e}$.
Suppose $w_{i,j} \in \mathbb{S}$ for each $i,j \in [j]$ (where we set $w_{i,j}:=\bar{0}=\infty$ if $(i,j)\notin E$), that $t \in [n]$ is a vertex and we want to evaluate the following ``$k$-path polynomial'' (which equals the minimum weight of a simple path on $k$ vertices by the choice of $\bS$):
\begin{equation}\label{eq:kppoly}
	\sum_{\substack{v_1,\ldots,v_k \in [n] \\  \forall i\neq j: v_i \neq v_j}} \prod_{i=1}^{k-1} w_{v_i,v_{i+1}}.
\end{equation}
For each $t \in [n]$ and $p \leq k$ define a vector $\vec{a} \in \bS^{2^{[n]}}$ as follows:
\[
\vec{a}_{t,p} = \sum_{\substack{v_1,\ldots,v_p \in [n] \\ v_p=t }} \left( \prod_{i=1}^{p-1} w_{v_i,v_{i+1}}\cdot\vec{C}_{n,v_i}\right)\vec{C}_{n,t}.
\]
So $\vec{a}_{t,k}[A]$ equals the minimum total weight of a simple path on $k$ vertices that ends at $t$ and visits all vertices in $A$ exactly once (and no other vertices).
Note that~\eqref{eq:kppoly} is equal to $(\vec{a}_{t,k}\cdot\vec{D}_{n,k})[\emptyset]=\sum_{A \subseteq \binom{[n]}{\leq k}}\vec{a}[A]$, and (as we show in Section~\ref{subsec:kpathformal}) we have the following recurrence for $p >0$:
\begin{align}
	\vec{a}_{t,p} 
	&= \sum_{u\in [n]} w_{u,t}\cdot \vec{a}_{u,p-1} \cdot \vec{C}_{n,t} \nonumber \\
	\intertext{Multiplying both sides with $\vec{D}_{n,k}$ and using $\vec{C}_{n,e} \cdot  \vec{D} = \vec{D}\cdot  \vec{C}^\intercal_{n,e}$, we can express $\vec{a}_{t,p}\cdot\vec{D}_{n,k}$ in terms of $\vec{a}_{u,p-1}\cdot\vec{D}_{n,k}$:}
	\vec{a}_{t,p}\cdot \vec{D}_{n,k} 
	&= \sum_{u\in [n]} w_{u,t}\cdot \vec{a}_{u,p-1}\cdot \vec{C}_{n,t}\cdot \vec{D}_{n,k}.\label{eq:simpath1}\\
	&= \sum_{u\in [n]} w_{u,t}\cdot  \vec{a}_{u,p-1}\cdot \vec{D}_{n,k} \cdot \vec{C}^\intercal_{n,t}.\label{eq:simpath2}
\end{align}
Since low rank allows for efficient handling of $\vec{a}_{u,p-1}\cdot\vec{D}_{n,k}$, this brings us to studying the rank of $\vec{D}_{n,k}$.
While the rank of $\mathbf{D}_{n,k}$ is large (see e.g.~\cite[Theorem 13.10]{DBLP:series/txtcs/Jukna11} for a proof over the field $\mathbb{F}_2$), the \emph{Boolean rank} of $\vec{D}_{n,k}$ is at most roughly $2^k$.\footnote{For an example, the factorization of $\vec{D}_{u,s}$ into $\vec{X},\vec{Y}$ in Section~\ref{sec:imp} is such a factorization of rank close to $2^k$.} This implies that, over any additively idempotent semiring, we have a rank $r$ factorization $\mathbf{D}_{n,k}=\vec{L}\cdot\vec{R}$ where $r$ is (close to) $2^k$.

Previous works for \textsc{Weighted Directed $k$-Path} use recurrences similar to~\eqref{eq:simpath1} and~\eqref{eq:simpath2}, but a bottleneck in these works is the computation of (an implicit representation of) $\vec{a}\cdot\vec{C}_{n,e}\cdot\vec{D}_{n,k}$ from (an implicit representation of ) $\vec{a}\cdot\vec{C}_{n,e}\cdot\vec{D}_{n,k}$, where $\vec{a}$ is the actual dynamic programming table. 
Roughly speaking, the previously fastest algorithms for this tasks assume $\vec{a}$ has small support and is explicitly given, and compute a representative set\footnote{See~\cite[Section 12.4]{DBLP:books/sp/CyganFKLMPPS15} for an introduction and~\cite{DBLP:journals/tcs/Tsur19b} for the currently fastest one. Note that in the literature these are often called \emph{`representative families for uniform matroids'}. We avoid the matroid terminology as it rules out many matrices $\mathbf{M}$, uses unnecessary jargon and yet matroid theory does not shed additional light on disjointness matrices.} of $\vec{a}\cdot\vec{C}_{n,e}$ of small support by using the mentioned Boolean rank $r$ factorization based on the following observation:\footnote{If $\bS$ is the min-sum semiring, one can define for each $j \in [r]$ an entry $\vec{a'}[i]:=\vec{a}[i]$ where $i$ is such that $\vec{L}[i,j]=1$ and $\vec{a}[i]$. All other entries of $\vec{a'}$ are defined to be $\infty$ and hence $|\supp(\vec{a'})|\leq r$.} If $\vec{a},\vec{a'} \in \bS^{2^{[n]}}$ are such that $\vec{a'}\cdot\mathbf{L}=\vec{a}\cdot\vec{L}$, then $\vec{a'}\cdot\vec{D}_{n,k}=\vec{a}\cdot\vec{D}_{n,k}$ and hence $\vec{a'}$ is a representative set for $\vec{a}$. 

An important ingredient in such algorithms that compute representative sets that we will also use is what we call the \emph{splitting trick}, which uses splitters (see Definition~\ref{def:spl} and Lemma's~\ref{lem:per} and~\ref{lem:split}) to hash and partition the universe $[n]$ into $s=\sqrt{k}$ blocks in such a way that the sought path $P$ has $s$ vertices in each block. A random hash function and partitioning of the universe has this property with large enough probability, and splitters can be seen as a derandomization of this argument. This allows us to work with the Kronecker power $\vec{D}_{u,s}^{\otimes s}$ as compatibility matrix instead of $\vec{D}_{n,k}$, where $u:=k^2$.

\subsection{Intuition behind our proof for Theorem~\ref{thm:kpath}.}\label{subsec:intuition}
Our main innovation allowing us to prove Theorem~\ref{thm:kpath} is to compute an (implicit) representation of $\vec{a}\cdot\vec{C}_{n,e}\cdot\vec{D}_{n,k}$ from a given a (similarly implicit) representation of $\vec{a}\cdot\vec{C}_{n,e}\cdot\vec{D}_{n,k}$. See Theorem~\ref{thm:main-sr} for precise statements. From this, Theorem~\ref{thm:kpath} follows in a relatively standard way and hence we only focus on the mentioned key step. By using the splitting trick as mentioned in the previous paragraph, it can be shown that it suffices to quickly compute an (implicit) representation of $\vec{a}\cdot(\vec{I}^{\otimes i-1}\otimes \vec{C}_{u,e}\otimes \vec{I}^{\otimes s - i})\cdot\vec{D}^{\otimes s}_{u,s}$ from a (similarly implicitly) given representation of $\vec{a}\cdot\vec{D}^{\otimes s}_{u,s}$ (i.e., when we want to insert a vertex of the path that is in the $i$-th block). Here $\vec{I}$ denotes the identity matrix of appropriate dimension.

Our starting point is a simple linear algebraic idea that works under the following (invalid!) assumption: Suppose that $\vec{D}_{u,s}$ has a rank-$\ell$ factorization $\vec{D}_{u,s}=\vec{X}\cdot\vec{Y}$ over the reals for some $\ell \leq 2^s$. Then we can assume $\vec{Y}$ has rank $\ell$ and hence there exists a right inverse $\vec{Y}^{-1}_{\mathsf{right}}$ such that $\vec{Y}\cdot\vec{Y}^{-1}_{\mathsf{right}}=\vec{I}$. Suppose that we have a representation $\vec{b}=\vec{a}\cdot\vec{X}^{\otimes s}$ of $\vec{a}\cdot\vec{D}^{\otimes s}_{u,s}$, then given $\vec{b}$ we can compute a representation of $\vec{a}\cdot(\vec{I}^{\otimes i-1}\otimes \vec{C}_{u,e}\otimes \vec{I}^{\otimes s - i})$ of the same type as follows:
\[
\begin{aligned}
	\vec{b}\cdot(\vec{I}^{\otimes i-1}\otimes \vec{Y}\cdot\vec{C}^\intercal_{u,e}\cdot\vec{Y}^{-1}_{\mathsf{right}}\otimes \vec{I}^{\otimes s-i}) &= \vec{a}\cdot (\vec{X}^{\otimes i-1} \otimes \vec{D}_{u,s}\cdot\vec{C}^\intercal_{u,e}\cdot\vec{Y}^{-1}_{\mathsf{right}}\otimes\vec{X}^{\otimes s-i})\\
	&=\vec{a}\cdot (\vec{X}^{\otimes i-1} \otimes \vec{C}_{u,e}\cdot\vec{D}_{u,s}\cdot\vec{Y}^{-1}_{\mathsf{right}}\otimes\vec{X}^{\otimes s-i})\\
	&=\vec{a}\cdot(\vec{I}^{\otimes i-1}\otimes \vec{C}_{u,e}\otimes \vec{I}^{\otimes s - i-1})\cdot\vec{X}^{\otimes s}.
\end{aligned}
\]
Since the vector $\vec{b}\cdot(\vec{I}^{\otimes i-1}\otimes \vec{Y}\cdot\vec{C}^\intercal_{u,e}\cdot\vec{Y}^{-1}_{\mathsf{right}}\otimes \vec{I}^{\otimes s-i})$ can be computed from the vector $\vec{b}$ in a direct way in $2^{s(s+1)}=2^{k+\sqrt{k}}$ time this gives a simple way to prove Theorem~\ref{thm:main-sr} (but, under an invalid assumption).

To make this idea work without the invalid assumption that $\vec{D}_{u,s}$ has small rank over the reals, we use the fact that $\vec{D}_{u,s}$ has small Boolean rank in combination with another type of ``inversion'' (Lemma~\ref{lem:sisr}) which heavily relies on a partial order induced by $\bS$ and is applied to $\vec{X}$ instead of $\vec{Y}$. In the above setting, our actual approach can be thought of as follows:
\begin{enumerate}
	\item Compute the \emph{unique minimal} vector $\vec{a}^*_{\mathsf{ext}}$ such that $\vec{b} \preceq_\bS \vec{a}^*_{\mathsf{ext}}(\mathbf{I}^{\otimes i-1}\otimes\vec{X}\otimes\vec{I}^{\otimes s-i})$.
	\item Output $\vec{a}^*_{\mathsf{ext}}\cdot(\vec{I}^{\otimes i-1}\otimes \vec{C}_{u,e}\cdot\vec{X}\otimes \vec{I}^{\otimes s - i})$
\end{enumerate}

Since we do not have additive inverses in the min-plus semiring, the inversion from Lemma~\ref{lem:sisr} is lossy and could lead to inconsistent data: If we first insert an element in the first block and apply ``inversion'' here to get a vector $\vec{b'}$ we cannot guarantee that $\vec{b}=\vec{a}^*_{\mathsf{ext}}(\vec{I}\otimes\vec{X}\otimes\vec{I}^{\otimes s-3})$ for some $\vec{a}^*_{\mathsf{ext}}$ and hence inversion may fail when we try to insert an element in the second block.
Also, the matrices $\vec{X}$ and $\vec{C}_{u,e}$ do not commute and hence it seems hard to recover e.g. $\vec{a}\cdot(\vec{C}_{u,e}\otimes \vec{I})$ from $\vec{a}\cdot\vec{X}^{\otimes 2}$ by only consider the ``slices'' induced by the second block separately.

However, the inversion is only lossy in one direction since in Step 1. there is always a unique \emph{minimal} solution, implying that the actual original solution is above $\vec{a}^*_{\mathsf{ext}}$.
Moreover, $\vec{C}_{u,e}$ and $\vec{D}_{u,s}$ \emph{do} commute. This allows us to define an upper bound (based on the minimality property of inversion) and a lower bound (that actually compares images of $\vec{D}_{u,s}$, which \emph{does} commute with $\vec{C}_{u,e}$) invariant that can be simultaneously maintained. See Definition~\ref{def:rep} for these invariant statements.

\subsection{The algorithmic tool of ``Dynamic Representative Sets'' and its applications beyond \textsc{$k$-Path}}
While it is convenient and tempting to restrict attention to the \textsc{Weighted Directed $k$-Path} problem,our main innovation as discussed above is quite general and hence it is worthwhile to relate it to Question~\ref{oqa}. Indeed, it still does not answer Question~\ref{oqa} for the special case of $\mathbf{M}$ being disjointness matrices $\mathbf{D}_{n,k}$ (for which representative sets of support size $2^k$ exist). However, the above method allows us to obtain the arguably "next best" thing: A data structure that supports all relevant operations and maintains some sort of representative set with small overhead in runtime.
We formalize this (in Theorem~\ref{thm:main-sr}) in what we call a \emph{dynamic representative set}.
In its most basic form (i.e., $\bS$ being the Boolean semiring $(\{0,1\},\vee,\wedge,0,1)$), it is given two parameters $0< k < n$ and allows us to maintain a representation of a family $\cF$ of subsets of $\{1,\ldots,n\}$. It supports basic update operations (unioning of two families, element convolution) and a query operation that determines for a set $B \subseteq \{1,\ldots,n\}$ whether there is a set $A \in \cF$ of size at most $k-|B|$ such that $A$ and $B$ are disjoint. After $2^{k+O(\sqrt{k}\log^2k)}n \log n$ preprocessing time, all operations use $2^{k+O(\sqrt{k}\log^2k)}\log n$ time.

We believe this formalism is easy to use as black box for speeding up relevant dynamic programming algorithms of which the compatibility matrices are disjointness matrices.
We leave it as a challenging open question whether such dynamic representative sets can also be maintained for other matrices $\mathbf{M}$ to improve e.g. the three applications mentioned below Question~\ref{oqa}.

To relate this more directly to the existing literature, we show how to use dynamic representative sets to get a fast deterministic algorithm for the following problem (see Subsection~\ref{subsec:mul} for definitions):

\defproblemu{\textsc{Skewed Multilinear Monomial Summation over $\bS$}}{A $d$-skewed arithmetic circuit $C$ over an additively idempotent semiring $\bS$ in commutative indeterminates $x_1,\ldots,x_n$ that computes a polynomial $C(x_1,\ldots,x_n)=\sum_{A \subseteq [n]}c_A$, and an integer $k$.}{The value $\sum_{A \in \binom{[n]}{k}} c_A$.}

\begin{theorem}\label{thm:mul}
	There is a deterministic algorithm for \textnormal{\textsc{Skewed Multilinear Monomial Summation Over $\bS$}} that uses $2^{k+O(\sqrt{k}\log^2k)}(n+d\cdot |C|)\log n$ time and arithmetic (i.e., additions, multiplications and $\mathsf{lcu}$ operations) operations in $\bS$.
\end{theorem}
This result improves over a plethora of previous results via known reductions. 
For example, Theorem~\ref{thm:mul} implies a deterministic $n^{O(t)}2^{k+O(\sqrt{k}\log^2 k)}$ time algorithm for the Subgraph Isomorphism problem. In this problem one is given an undirected $n$ vertex graph $G$ and another undirected graph $P$ along with a tree decomposition of $P$ of width $t$. Then in the mentioned running time we can determine whether $P$ occurs as a subgraph of $G$. 
See e.g.~the discussion in~\cite{DBLP:journals/jacm/FominLPS16,DBLP:conf/soda/LokshtanovSZ21} for details.

Other examples of problems for which Theorem~\ref{thm:mul} gives faster deterministic algorithms are e.g. $t$-Dominating (see~\cite{DBLP:journals/talg/KoutisW16}) and Weighted $3$-set packing (see~\cite{DBLP:journals/toct/Zehavi23}).

For the special case where $\bS$ is the Boolean semiring, \textnormal{\textsc{Skewed Multilinear Monomial Summation Over $\bS$}} reduces to the \textnormal{\textsc{Skewed Multilinear Monomial Detection}} problem, which can be solved with a randomized algorithm in $O^*(2^k)$ time~\cite{DBLP:journals/talg/KoutisW16}. Our result can be seen as a derandomization of that result (at the cost the lower order order term $2^{O(\sqrt{k}\log^2k)}$ in the run time). While studying the problem over idempotent semirings may not be so natural at first sight, we believe this is an interesting generalization because:
\begin{enumerate}
	\item Algebra has played a main role in the algorithmic literature on Multilinear Detection problems (examples are the Group Algebra approach from~\cite{DBLP:journals/talg/KoutisW16}, Exterior Algebra~\cite{DBLP:conf/stoc/BrandDH18, DBLP:journals/jacm/FominLPS16}, Waring Rank~\cite{DBLP:conf/focs/Pratt19}).
	\item The extension from e.g. Boolean/Min-plus semirings to general idempotent semirings is non-trivial and relies on the existence of an induced order of the elements of $\bS$ with lattice properties (used in Lemma~\ref{lem:sisr}). Moreover, the idempotency requirement is probably needed since an extension of Theorem~\ref{thm:mul} to arbitrary semirings would imply an algorithm for counting (or even only computing the parity of) the number of simple paths on $k$ vertices that is unlikely to exist (by e.g. \#W[1]-hardness and $\oplus$W-hardness due to~\cite{DBLP:journals/siamcomp/FlumG04,DBLP:conf/esa/CurticapeanDH21}).
\end{enumerate} 
\subsection{Related work and comparison with our approach.}
The earliest application of representative sets was for the $k$-\textsc{Path} problem in~\cite{MONIEN1985239}. Other notable publications using it include~\cite{DBLP:journals/algorithmica/AlonYZ97,DBLP:journals/tcs/Marx09,DBLP:journals/talg/KratschW14,DBLP:journals/iandc/BodlaenderCKN15,DBLP:conf/soda/FominLS14,DBLP:conf/esa/Zehavi15,DBLP:journals/jcss/ShachnaiZ16,DBLP:journals/tcs/Tsur19b}.

Various extensions of representative sets have been studied in the literature, including extensions to approximate counting~\cite{DBLP:conf/soda/LokshtanovSZ21} and approximation algorithm~\cite{DBLP:journals/tcs/Zehavi16}.

Our approach builds mainly on classic tools used for the $k$-path problem such as splitters, universal sets and representative sets. 
The first use of splitters as it is done Lemma~\ref{lem:per} already occurred in~\cite{DBLP:journals/jacm/AlonYZ95}.
In particular, universal sets have been used to obtain deterministic algorithms for \textsc{Weighted $k$-Path} previously in~\cite{doi:10.1137/080716475}, and the use of ``skewed'' universal sets to compute representative sets was introduced in~\cite{DBLP:conf/soda/FominLS14,DBLP:journals/jacm/FominLPS16}. 
The aforementioned ``splitting trick' to 'use of Lemma~\ref{lem:split} to split the solution in $\sqrt{k}$ parts occurred also in previous work, see e.g.~\cite{DBLP:conf/soda/LokshtanovSZ21}.

\subsection*{Organization.}
The remainder of this paper is organized as follows:
In Section~\ref{sec:prelnot} we provide some preliminaries and notation.
In Section~\ref{sec:dynrep} we state our dynamic representative set framework and use it to prove Theorem~\ref{thm:kpath} and Theorem~\ref{thm:mul}.
In Section~\ref{sec:imp} we present the formal proof of how the dynamic representative set is achieved.

\section{Preliminaries and notation}\label{sec:prelnot}
For an integer $n$, we let $[n]$ denote the set $\{1,\ldots,n\}$. If $f:A\rightarrow B$ is a function, it is extended in the natural way to a function $f:2^A\rightarrow 2^B$ (i.e., if $X\subseteq A$, then $f(X)$ denotes $\{f(x): x \in X\}$).
\paragraph{Idempotent Semirings.}
To state our general results we need to recall some basic algebra:
\begin{definition}
	A \emph{semiring} is a quintuple $\mathbb{S}=(S,+,*,\bar{0},\bar{1})$ such that
	\begin{enumerate}
		\item $(S,+,\bar{0})$ is a commutative monoid: $(a+b)+c=a+(b+c)$ and $a+b=b+a$ for all $a,b,c \in S$, and $a+\bar{0}=a$ for all $a \in A$,
		\item $(S,*,\bar{1})$ is a monoid: $(a*b)*c=a*(b*c)$ for all $a,b,c \in S$ and $a*\bar{1}=\bar{1}*a=a$ for all $a \in A$,
		\item The multiplication operator $*$ distributes over the addition operator $+$: for every $a,b,c \in S$:
		\[
			(a+b)*c = a*c+b*c \qquad \text{ and } \qquad a*(b+c) = a*b+a*c.
		\]
		\item The additive identity $\bar{0}$ is an annihilator for multiplication: $a*\bar{0}=\bar{0}*a=\bar{0}$ for every $a \in S$.
	\end{enumerate} 
\end{definition}
The two semirings that are most relevant for this paper are the Boolean semiring $(\{0,1\},\vee,\wedge,0,1)$ (useful for solving decision problems), and the min-plus semiring $(\mathbb{Z}_{\geq 0} \cup \infty,\min,+,\infty,0)$ (useful for optimization problems.). Since the min-plus semiring is infinite, we emulate it with the \emph{$M$-capped min-plus semiring}, which is defined as the semiring $(\{0,\ldots,M\} \cup \infty,\min,+,\infty,0)$ with all integers larger than $M$ being mapped to $\infty$. It is easily verified that this is still a semiring. All these semirings have the following property:

\begin{definition}
	A semiring $(S,+,*,\bar{0},\bar{1})$ is \emph{additively idempotent} if $a+a=a$ for every $a \in A$.
\end{definition}

Assuming $\mathbb{S}=(S,+,*,\bar{0},\bar{1})$ is additively idempotent, we define a relation $\preceq_\mathbb{S}$ on $S$ as follows: $a \preceq_{\mathbb{S}} b$ if and only if $a+b=a$.
It is well known that $(S,\preceq_{\mathbb{S}})$ is a lower semilattice: It is easily seen to be a partial ordered set,\footnote{Recall, an ordering $\preceq$ is a partial ordered set if for every $x,y,z$ it holds that $x\preceq x$ (\emph{reflexivity}), $x \leq y$ and $y\leq x$ implies $x=y$ (\emph{antisymmetry}), and $x\leq y$ together with $y \leq z$ implies $x\leq z$ (\emph{transitivity}).} and the unique greatest lower bound of $a,b$ is $a+b$. To see this, suppose $c \leq a$ and $c\leq b$. Then $c+a=c+b=c$, and therefore $c+(a+b) =c$ which means that $c \preceq_{\mathbb{S}} a+b$.
Moreover, $a \preceq_\bS \bar{0}$ for every $a \in S$. This implies that, if $S$ is finite, then $\preceq_\bS$ is in fact a lattice: Given elements $x_1,\ldots,x_\ell$ we can define a \emph{least common upper bound}  
\[
	\mathrm{lcu}(x_1,\ldots,x_\ell) := \sum_{y \text{ s.t. } \forall i \in [\ell]:  x_i \preceq_\bS y} y.
\]
We assume that this $\mathrm{lcu}$ is efficiently computable and call the addition, multiplication and $\mathsf{lcu}$ operation \emph{arithmetic operations in $\bS$}. Note that in the ($M$-capped) min-plus semi-ring, $\mathsf{lcu}(x_1,\ldots,x_\ell)$ equals $\max\{x_1,\ldots,x_\ell\}$.

\paragraph{Linear Algebra.}
Vectors and matrices are denoted in boldface. They will be interchangeably be indexed both with integers and sets (assuming a fixed ordering of the set).

\begin{definition} If $\vec{A} \in \bS^{n\times m}$ and $\vec{B} \in \bS^{n'\times m'}$, their \emph{Kronecker product} $\vec{A} \otimes \vec{B} \in \bS^{([n]\times[n'])\times([m]\times[m'])}$ is defined as follows:
	\[
		(\vec{A} \otimes \vec{B})[(i,j),(i',j')]:= \vec{A}[i,i']\cdot \vec{B}[j,j'].
	\]
If $s \in\mathbb{N}_{> 0}$, we denote $\vec{A}^{\otimes s}$ for the $s$-fold Kronecker product $\vec{A}\otimes\cdots\otimes\vec{A}$.
\end{definition}
A basic property of this product used frequently and implicitly throughout this paper is the \emph{Mixed Product Property}:
\[
(\vec{A}\cdot \vec{B}) \otimes (\vec{C}\cdot \vec{D}) = (\vec{A}\otimes\vec{C}) \cdot (\vec{B}\otimes\vec{D}).
\]
Matrices and vectors are denoted in boldface throughout this paper. By default, vectors denote row vectors (i.e. $(1 \times n)$-matrices). We use $\vec{A}^\intercal$ to denote the transpose of matrix $\vec{A}$.

Throughout this paper, the semiring $\bS$ will be implicitly extended to a matrix semiring: If $\vec{A} \in \bS^{n\times m}$ and $\vec{B} \in \bS^{m\times n}$, then $(\vec{A}\cdot\vec{B})[i,k]=\sum_{j=1}^m\vec{A}[i,j]\cdot\vec{B}[j,k]$, where the summation and multiplication are in $\bS$. Likewise, the ordering $\preceq_\bS$ is extended to vectors and matrices in the natural way: If $\vec{a}, \vec{b} \in \bS^n$, then $\vec{a} \preceq_\bS \vec{b}$ denotes that $\vec{a}[i] \preceq_\bS \vec{b}[i]$ for every $i$.

\paragraph{Splitters and Universal Families.}
\newcommand{\splt}{\ensuremath{\mathrm{split}}}

Following the approach of previous work on the \textsc{Weighted Directed} $k$-\textsc{Path} problem and related problems, we rely on a number of pseudo-random families, which we will outline now.

\begin{definition}\label{def:con}
For a function $f : [a] \rightarrow [b]$, \emph{constructing $f$} means computing an array that stores $f(x)$ for every $x\in [a]$. In the word RAM model with word size at least $a+b$ this means that, after constructing $f$, we can evaluate $f$ on any argument in $O(1)$ time.
\end{definition}

\begin{definition}
	Let $S \in \binom{[n]}{k}$ and let $h: [n] \rightarrow [\ell]$. We say that \emph{$h$ splits $S$} if for some $j$
	\[
		|h^{-1}(i) \cap S| =
		\begin{cases}
			\lceil k/\ell\rceil,  & \text{ if } i \leq j,\\
			\lfloor k/ \ell \rfloor& \text{otherwise}.
		\end{cases} 
	\]
\end{definition}

\begin{definition}\label{def:spl}
An \emph{$(n,k,\ell)$-splitter} $\cH$ is a family of functions from $[n]$ to $[\ell]$ such that for every set $S \subseteq [n]$ of size $k$ there exists a function $h \in 
\cH$ that splits $S$.
\end{definition}
We will combine in a standard way the following two constructions of splitters:

\begin{lemma}[\cite{DBLP:journals/jacm/AlonYZ95}]\label{lem:per}
For any $n,k \geq 1$, one can construct a $(n,k,k^2)$-splitter of size $k^{O(1)}\log n$ in time $k^{O(1)} n \log n$. 
\end{lemma}

\begin{lemma}[\cite{DBLP:conf/focs/NaorSS95}, Theorem 3(i)]\label{lem:split} For $\ell = O(\sqrt{k})$ one can construct an $(n,k,\ell)$-splitter of size $k^{O(\ell)}$ in time $k^{O(\ell)}n^{O(1)}$.
\end{lemma}

Like previous related deterministic algorithms (e.g.~\cite{doi:10.1137/080716475}), our algorithms also rely on the following:

\begin{definition}\label{def:uni}
	Let $u$ and $s$ be an integers. A $(u,s)$-universal family is a family $\cF \subseteq 2^{[u]}$ such that for every $X \in \binom{[u]}{\leq s}$ we have that $\{X \cap S\mid S \in \cF \}=2^X$.
\end{definition}
\begin{lemma}[\cite{DBLP:conf/focs/NaorSS95}, Theorem 6]\label{lem:unifam}
	A $(u,s)$-universal family of size $2^ss^{O(\log s)}\log u$ can be constructed in $2^ss^{O(\log s)}\log u$ time.
\end{lemma}

\section{Dynamic Representative Sets: Statement and Applications}\label{sec:dynrep}
We will now formally state our dynamic representative set tool. Following the discussion of Sections~\ref{sec:repsets} and~\ref{subsec:kpath} we aim to efficiently represent $\vec{a}\cdot \vec{D}_{n,k}$ in such a way that a representation of $\vec{a}\cdot\vec{C}_{n,e}\cdot\vec{D}_{n,k}$ can be quickly computed from a representation of $\vec{a}\cdot \vec{D}_{n,k}$. 
We will show this in Theorem~\ref{thm:main-sr} (and Lemma~\ref{lem:repcon}) for a notion of ``represents'' (outlined in Definition~\ref{def:rep}) that is slightly stronger than just being able to recover $\vec{a}\cdot\vec{D}_{n,k}$ (as shown in Observation~\ref{obs:query}).

Fix an additively idempotent semiring $\mathbb{S}=(S, +, \cdot,\bar{0},\bar{1})$.
We encourage reader to assume $\bS$ is the min-sum semi-ring $(\mathbb{Z}_\geq,\min,+,\infty,0)$ during their first read. In this case the order $\preceq_\bS$ (defined in Section~\ref{sec:prelnot}) reduces to the natural order on integers and $\mathsf{lcu}$ reduces to $\max$. In this section, $+,\cdot$ refers to the operations of $\mathbb{S}$. Likewise, the Iverson bracket notation $\llbracket b \rrbracket$, for a boolean $b$, indicates $\bar{0}$ if $b$ is false, and $\bar{1}$ if $b$ is true.

Throughout this section, we fix two integers $k,n$ with $k\leq n$. Recall from~Subsection~\ref{subsec:kpath} that	
\[
\vec{C}_{n,e}[A,B] := \llbracket e \notin A \text{ and } A \cup \{e\} = B \rrbracket \qquad \text{and}\qquad
\vec{D}_{n,k}[A,B] :=	\llbracket A \cap B = \emptyset \text{ and } |A \cup B|\leq k \rrbracket.
\]	
A crucial property is that these matrices commute in the following sense:
\begin{lemma}\label{lem:comm} 
	If $e \in U$, then $\vec{C}_{n,e} \cdot  \vec{D}_{n,k} = \vec{D}_{n,k}\cdot  \vec{C}_{n,e}^\intercal$.
\end{lemma}
\begin{proof}
	By expanding the matrix multiplications, we have
	\[
	\begin{aligned}
		(\vec{C}_{n,e}\cdot \vec{D}_{n,k})[A,B] &= \begin{cases}
			\mathrlap{\vec{D}_{n,k}[A \cup \{e\},B],}\hphantom{\vec{D}_{n,k}[A,B \cup \{e\}],} & \text{if $e \notin A$},\\
			\bar{0}, & \text{otherwise},
		\end{cases}\\
		(\vec{D}_{n,k}\cdot \vec{C}^\intercal_{n,e})[A,B] &= \begin{cases}
			\vec{D}_{n,k}[A,B \cup \{e\}], & \text{if $e \notin B$},\\
			\bar{0}, & \text{otherwise}.
		\end{cases}	
	\end{aligned}
	\]
	Hence, both terms are equal to $\bar{1}$ if $e \not\in A\cup B$, $A \cap B = \emptyset$ and $|A \cup B|\leq k-1$, and otherwise both terms are equal to $\bar{0}$.
\end{proof}
We now state our main tool, discuss its utility in this section and postpone its proof to  Section~\ref{sec:imp}.
\begin{theorem}[Dynamic Representative set]\label{thm:main-sr}
	Let $\mathbb{S}=(S,+,*,\bar{0},\bar{1})$ be an additively idempotent semiring and let $n,k$ be integers satisfying $k\leq n$. There exists a factorization  $\vec{D}_{n,k}=\vec{L}\cdot\vec{R}$ of rank $r \leq 2^{k+O(\sqrt{k}\cdot \log^2k)}\log n$ such that, after $2^{O(\sqrt{k}\cdot \log^2k)}n \log n$ preprocessing time:
	\begin{enumerate}
		\item[\hypertarget{P1}{\normalfont \textcolor{gray}{\textbf{(P1)}}}] Any entry of $\vec{L}$ and $\vec{R}$ can be computed in $k^{O(1)}$ time, assuming the subset of $[n]$ that indexes the row of $\vec{L}$ (or the column of $\vec{R}$) is given as sequence of its elements,\footnote{Entries of $\vec{L}$ and $\vec{R}$ indexed by subsets of more than $k$ elements actually always are equal to $\bar{0}$.}
		\item[\hypertarget{P2}{\normalfont \textcolor{gray}{\textbf{(P2)}}}] there is an algorithm $\mathsf{convolve}(\vec{b},e)$ that takes as input a vector $\vec{b} \in \mathbb{S}^{r}$ and an element $e \in [n]$, and it computes in $2^{k+O(\sqrt{k}\cdot \log^2 k)} \log n$ time and arithmetic operations in $\bS$ (i.e. addition, multiplication, $\mathsf{lcu}$) a vector $\vec{b'} \in \bS^r$ such that $\vec{b}\cdot\vec{R}\cdot\vec{C}^\intercal_{n,e} \preceq_\bS \vec{b'}\cdot\vec{R}$ and
		\[
		\text{for all $\vec{a} \in \bS^{2^{[n]}}$ such that } \vec{b} \preceq_\mathbb{S} \vec{a} \cdot \vec{L}, 
		\text{ it holds that } \vec{b'} \preceq_\mathbb{S} \vec{a}\cdot \vec{C}_{n,e} \cdot \vec{L}.
		\]
	\end{enumerate}
	 Here all time bounds are with respect to the Word RAM model with word size $\log n$.	 
\end{theorem}
\newcommand{\Pone}{\hyperlink{P1}{\normalfont \textcolor{gray}{\textbf{(P1)}}}}
\newcommand{\Ptwo}{\hyperlink{P2}{\normalfont \textcolor{gray}{\textbf{(P2)}}}}
\newcommand{\Pthree}{\hyperlink{P3}{\normalfont \textcolor{gray}{\textbf{(P3)}}}}
\newcommand{\Pfour}{\hyperlink{P4}{\normalfont \textcolor{gray}{\textbf{(P4)}}}}
\newcommand{\Pfive}{\hyperlink{P5}{\normalfont \textcolor{gray}{\textbf{(P5)}}}}
While the statement of Theorem~\ref{thm:main-sr} is convenient for its proof, we provide some simple definitions and observations to facilitate its use.
In their statements, we let integers $n,k,r$ and matrices $\vec{L},\vec{R}$ be as in Theorem~\ref{thm:main-sr}.
We first connect Theorem~\ref{thm:main-sr} to its name ``dynamic representative set'':
\begin{definition}\label{def:rep}
	A vector $\vec{b} \in \bS^r$ \emph{represents} a vector $\vec{a} \in \bS^{2^{[n]}}$ if $\vec{a}\cdot\vec{D}_{n,k} \preceq_\bS \vec{b}\cdot\vec{R}$ and $\vec{b} \preceq_\bS \vec{a}\cdot\vec{L}$.
\end{definition}
Note $\vec{b} \preceq_\bS \vec{a}\cdot\vec{L}$ implies $\vec{b}\cdot\vec{R} \preceq_\bS \vec{a}\cdot\vec{L}\cdot\vec{R}=\vec{a}\cdot\vec{D}_{n,k}$. Thus, by anti-symmetry of $\preceq_\bS$ we have:
\begin{observation}\label{obs:query}
	If $\vec{b}$ represents $\vec{a}$ then $\vec{a}\cdot\vec{D}_{n,k}=\vec{b}\cdot\vec{R}$.
\end{observation}
The utility of Theorem~\ref{thm:main-sr} lies in that it preserves representation:
\begin{lemma}\label{lem:repcon}
	If $\vec{b}$ represents $\vec{a}$, then $\mathsf{convolve}(\vec{b},e)$ represents $\vec{a}\cdot\vec{C}_{n,e}$. 
\end{lemma}
\begin{proof}
	Let $\vec{b'}=\mathsf{convolve}(\vec{b},e)$. 
	For the first condition of Definition~\ref{def:rep}, note that by Lemma~\ref{lem:comm} we have that
	\[
	\vec{a}\cdot\vec{C}_{n,e}\cdot\vec{D}_{n,k}=\vec{a}\cdot\vec{D}_{n,k}\cdot\vec{C}^\intercal_{n,e} \preceq_{\bS} \vec{b}\cdot\vec{R}\cdot\vec{C}^\intercal_{n,e}\preceq_\bS \vec{b'}\cdot\vec{R},
	\]
	where the two inequalities follow from respectively $\vec{b}$ representing $\vec{a}$, and \Ptwo{} of Theorem~\ref{thm:main-sr}.
	
	For the second condition of Definition~\ref{def:rep}, note that since $\vec{b} \preceq \vec{a}\cdot\vec{L}$, \Ptwo{} of Theorem~\ref{thm:main-sr} implies $\vec{b'}\preceq_\bS \vec{a}\cdot\vec{C}_{n,e}\cdot\vec{L}$.
\end{proof}

The next two observations follow directly from Definition~\ref{def:rep} and the axioms of linear algebra:
\begin{observation}\label{obs:lin}
	If $\vec{b}$ represents $\vec{a}$ and $\vec{b'}$ represents $\vec{a'}$, then $\vec{b}+\vec{b'}$ represents $\vec{a}+\vec{a'}$.
\end{observation}
\begin{observation}\label{obs:scal}
	If $\lambda \in \bS$ and $\vec{b}$ represents $\vec{a}$ then $\lambda\cdot \vec{b}$ represents $\lambda\cdot\vec{a}$ and $\vec{b}\cdot\lambda$ represents $\vec{a}\cdot\lambda$.
\end{observation}
The final observation computes a representation of an initial vector: 
\begin{observation}\label{obs:init}
	We can compute in $2^{k+O(\sqrt{k}\cdot \log^2k)}n \log n$ time the vector $\vec{b}_{\mathsf{init}}:=\vec{a}_{\mathsf{init}}\cdot\vec{L}$, where $\vec{a}_{\mathsf{init}}$ denotes the vector that satisfies $\vec{a}_{\mathsf{init}}[A]= \llbracket A = \emptyset \rrbracket$, and $\vec{b}_{\mathsf{init}}$ represents $\vec{a}_{\mathsf{init}}$.
\end{observation}

It follows from \Pone{} of Theorem~\ref{thm:main-sr}, which implies that indeed  $\mathbf{b}_{\mathsf{init}}[r]=\mathbf{L}[\emptyset,r]$ can be computed in the claimed time, and $\vec{b}_{\mathsf{init}}$ represents $\vec{a}_{\mathsf{init}}$ because $\vec{L}\cdot\vec{R}=\vec{D}_{n,k}$.

\subsection{Application to \textsc{Weighted Directed} $k$-\textsc{Path}.}\label{subsec:kpathformal}
We now show how Theorem~\ref{thm:main-sr} implies Theorem~\ref{thm:kpath}. This is quite standard and a small adaptation of e.g.~\cite[Chapter 12]{DBLP:books/sp/CyganFKLMPPS15} and also introduced already in Subsection~\ref{subsec:kpath}. It will also be generalized in the next subsection (since Theorem~\ref{thm:mul} implies Theorem~\ref{thm:kpath} by a known reduction), but it is instructive to see a separate proof of Theorem~\ref{thm:kpath}.

Given is an instance $G=([n],E)$, $\{w_{i,j}\}_{(i,j) \in E}$ and $k$ of \textsc{Weighted Directed $k$-Path}.
We will apply Theorem~\ref{thm:main-sr} with $\cS$ being the $W$-capped min-plus semiring where $W:=k\cdot\max_{(i,j) \in E}w_{i,j}$, but stick to the general semiring notation for convenience.
Recall that we defined in Section~\ref{subsec:kpath}, for each $t \in [n]$ and $p \leq k$, a vector $\vec{a}_{t,p} \in \bS^{2^{[n]}}$ as follows (where we set $w_{i,j}:=\bar{0}=\infty$ if $(i,j)\notin E$):
\begin{align}
\vec{a}_{t,p} &:= \sum_{\substack{v_1,\ldots,v_p \in [n] \\ v_p=t }} \left( \prod_{i=1}^{p-1} w_{v_i,v_{i+1}}\cdot\vec{C}_{n,v_i}\right)\vec{C}_{n,t}.\label{eq:path}
\intertext{By the definition of $\vec{C}_{n,e}$, we have that for each $A$ of size $p$:}
\vec{a}_{t,p}[A]&=\sum_{\substack{v_1,\ldots,v_{p} \in [n]\\ v_p=t\\ \{v_1,\ldots,v_{p}\} = A }} \prod_{i=1}^{p-1} w_{v_i,v_{i+1}}.\nonumber
\end{align}
Hence, the answer to the instance of \textsc{Weighted Directed $k$-Path} is $\sum_{A \subseteq \binom{[n]}{\leq k}}\vec{a}[A]=(\vec{a}_{t,k}\cdot\vec{D}_{n,k})[\emptyset]$. Letting $\vec{a}_{\mathsf{init}}$ denote the vector with $\vec{a}_{\mathsf{init}}[A]=\llbracket A=\emptyset\rrbracket$ for all $A \subseteq [n]$, we have that $\vec{a}_{t,p}$ satisfies
\begin{align*}
	\setcounter{equation}{\theequation+1}
	\hspace{8em}\vec{a}_{t,p} &=
	\begin{cases}
		\mathrlap{\vec{a}_{\mathsf{init}}\cdot \vec{C}_{n,t}}\hphantom{\displaystyle \sum_{u \in N^-(t)} w_{u,t}\cdot \mathsf{convolve}(\vec{b}_{u,p-1},t)}& \text{if } p=1,\hspace{9.7em}\hypertarget{Z2}{\textrm{(\theequation a)}} \\
		\displaystyle \sum_{u \in N^-(t)} w_{u,t}\cdot \vec{a}_{u,p-1} \cdot \vec{C}_{n,t} & \text{otherwise.}\hfill\hypertarget{Z3}{\textrm{(\theequation b)}}
	\end{cases}
	\intertext{To see~\mylink{Z2} note from~\eqref{eq:path} that $\vec{a}_{t,1}[A]$ is equal to $\bar{1}$ if $A=\{t\}$ and it is equal to $\bar{0}$ otherwise. To see~\mylink{Z3}, note that in in both~\mylink{Z3} and~\eqref{eq:path} we sum over all options $u$ for $v_{p-1}$ and for a given choice multiply with $w_{u,t}=w_{v_{p-1},v_p}$ and $\vec{C}_{n,t}$. The algorithm implementing Theorem~\ref{thm:kpath} computes vectors $\vec{b}_{t,p}$ for $t \in V(G)$ and $p=1,\ldots,k'$ as follows, where $\vec{b}_{\mathsf{init}}$ is the vector from Observation~\ref{obs:init}:
	}
	\vec{b}_{t,p} &:=
	\begin{cases}\setcounter{equation}{\theequation+1}
		\mathsf{convolve}(\vec{b}_{\mathsf{init}}, t) & \text{if } p=1, \\
		\displaystyle \sum_{u \in N^-(t)} w_{u,t}\cdot \mathsf{convolve}(\vec{b}_{u,p-1},t) & \text{otherwise.}
	\end{cases}\setcounter{equation}{\theequation-1}
\end{align*}
By Theorem~\ref{thm:main-sr} and Observation~\ref{obs:init}, we can compute all $\vec{b}_{t,p}$ using a total of $2^{k+O(\sqrt{k}\cdot \log^2k)}n \log n$ time and arithmetic operations in $\bS$. Note that arithmetic operations in the $M$-capped min-plus semiring take $O(1)$ time in the word RAM model with word size $\log M$ (since $\mathsf{lcu}$ reduces to $\max$ here).

Now we argue by induction on $p=1,\ldots,k+1$ that $\vec{b}_{t,p}$ represents $\vec{a}_{t,p}$ for every $t$ and $p$: For $p=1$, this follows from Observation~\ref{obs:init} and \Ptwo{} of Theorem~\ref{thm:main-sr} and Lemma~\ref{lem:repcon}. Moreover, the case $p>1$ follows from applying Observation~\ref{obs:lin} for the summation over $u \in N^{-}(t)$, Observation~\ref{obs:scal} for the scalar multiplication with $w_{u,t}$ and \Ptwo{} of Theorem~\ref{thm:main-sr} and Lemma~\ref{lem:repcon} for the element convolution with $t$ in combination with the induction hypothesis.

Hence, $\vec{b}_{t,k}$ represents $\vec{a}_{t,k}$ and by Observation~\ref{obs:query} we have that $(\vec{a}_{t,k}\cdot\vec{D}_{n,k})[\emptyset]=(\vec{b}_{t,k}\cdot\vec{R})[\emptyset]$.
By \Pone{} of Theorem~\ref{thm:main-sr}, the minimum length of a $k$-path $(\sum_{t\in V}\vec{b}_{t,k}\cdot\vec{R})[\emptyset]$ can be computed in $2^{k+O(\sqrt{k}\cdot \log^2k)}n\log n$ time. This concludes the proof of Theorem~\ref{thm:kpath}.

\subsection{Application to \textsc{\textsc{Skewed Multilinear Monomial Summation over $\bS$}}.}\label{subsec:mul}
Let $\bS=(S,+,*,\bar{0},\bar{1})$ be an additively idempotent semiring and let $n,k$ be integers with $ k\leq n$.
\begin{definition}[A $d$-skewed arithmetic circuit $C$ over $\bS$]\label{def:skew}
A \emph{$d$-skewed arithmetic circuit $C$ over $\bS$} is a tuple $(D,\lambda,\pi_l,\pi_r)$, where $D=(V,E)$ is a directed acyclic graph and a labeling function $\lambda: V \rightarrow \{x_1,\ldots,x_n,*,+\} \cup S$ assigns indeterminates $x_1,\ldots,x_n$ or elements of $\bS$ to vertices of $D$ such that:
\begin{enumerate}
	\item If $g \in V$ is a source of $G$, $\lambda(g) \in \{x_1,\ldots,x_n\} \cup S$, and otherwise $\lambda(g) \in \{+,*\}$. 
	\item If $\lambda(g)=*$, then $g$ has exactly two in-neighbors, denoted $\pi_l(g)$ and $\pi_r(g)$.
	\item We associate with every vertex $g$ of $C$ a polynomial $P_g(x_1,\ldots,x_n)$ as follows:
	\begin{itemize}
		\item If $g$ is a leaf, $P_g=\lambda(g)$,
		\item If $g$ is a multiplication gate, then $P_g := P_{\pi_l(g)}\cdot P_{\pi_r(g)}$,
		\item If $g$ is an addition gate, then $P_g:=\sum_{u: (u,g) \in E} P_u$.
	\end{itemize}
	\item For each $g$ such that $\lambda(g)=*$, at least one of $P_{\pi_l(g)}$ and $P_{\pi_r(g)}$ contains at most $d$ monomials.\footnote{We note there exist different definitions of $d$-skewed circuits. We use the one from~\cite{DBLP:conf/soda/LokshtanovSZ21}.}
	\item $C$ has a unique sink $s$.
\end{enumerate}
Here we assume that the indeterminates $x_1,\ldots,x_n$ commute with each other and with elements of $S$, that is $x_i\cdot x_j=x_j\cdot x_i$ and $s\cdot x_i=x_i\cdot s$ for all $i\neq j$ and $s\in S$. 
\end{definition}

We let $|C|$ denote the number of gates of $C$ (e.g., $|V|$). Because of the commutativity assumptions of Definition~\ref{def:skew}, any polynomial $P_g$ can be written as $\sum_{A \subseteq [n]}c_A\prod_{i\in A}x_i$, for some coefficients $c_A \in S$. In partilcular, we can uniquely address any multilinear monomial (e.g. a monomial in which each variable as exponent at most $1$) of $P_g$ with the set of variables occuring in it.

The \textsc{Skewed Multilinear Monomial Detection over $\bS$} problem is formally defined as follows: Given is a $d$-skewed circuit $C$ with unique sink $s$ that computes a polynomial $P_s=\sum_{A \subseteq [n]} c_{A}\prod_{i \in A}x_i$ and an integer $k$. The task is to compute $\sum_{A \in \binom{[n]}{k}}c_A$.

We use the following preprocessing step (that can be proved via simple dynamic programming) to simplify the multiplication gates by using the $d$-skewedness property:

\begin{lemma}[Proposition 4.1 in the full version\footnote{Which currently only seems accessible at~\href{https://sites.cs.ucsb.edu/~daniello/papers/countingPathsSODA2021.pdf}{https://sites.cs.ucsb.edu/~daniello/papers/countingPathsSODA2021.pdf}.} of~\cite{DBLP:conf/soda/LokshtanovSZ21}]
	On $d$-skewed circuits for any $d \in \mathbb{N}$, the following can be computed in time $O(|C|\cdot d)$ time: For every gate $g$ in $C$, if $P_g$ has at most $d$ monomials, an explicit representation $Q_g$ of the polynomial $P_g$ (as a list of all the at most $d$ monomials along with their coefficients; we will address $Q_g$ as the set of pairs consisting of a monomial along with its non-zero coefficients in $P_g$), and otherwise $Q_g := null$.
\end{lemma}

We now prove Theorem~\ref{thm:mul} using Theorem~\ref{thm:main-sr}.
For every gate $g$, with associated polynomial $P_g(x_1,\ldots,x_n)=\sum_{A \subseteq [n]} c_A \prod_{i\in A}x_i$, define
\[
	\vec{a}_{g,p}[A] :=
	\begin{cases}
		c_A & \text{if }|A|=p,\\
		\bar{0} &\text{otherwise}.
	\end{cases}	 
\]
We use Theorem~\ref{thm:main-sr} to compute for every integer $p=0,\ldots,k$ and every gate $g$ of $C$ (in a topological order of the underlying graph $D$) a vector $\vec{b}_{g,p}$ as follows:
If $P_g$ has at most $d$ monomials, use $Q_g$ to construct
\[
	\vec{b}_{g,p} := \sum_{(\{x_{e_1},\ldots, x_{e_p}\}, \lambda) \in Q_g} \lambda \cdot \mathsf{convolve}(\vec{b}_{\mathsf{init}},\{e_1,\cdots, e_p\}),
\]
where we let $\mathsf{convolve}(\vec{b},\{e_1,\ldots,e_p\})$ denote the vector obtained by applying $\mathsf{convolve}$ on $\vec{b}$ iteratively $p$ times with as second argument $e_1,\ldots,e_p$.
By Lemma~~\ref{lem:repcon} and Observations,~\ref{obs:lin},~\ref{obs:scal}, and~\ref{obs:init} we have that $\vec{b}_{g,p}$ represents $\vec{a}_{g,p}$.

Otherwise, if we are given $\vec{b}_{g_i,p}$ for all in-neighbors $g_1,\ldots,g_\ell$ of a gate $g$ and $0\leq p \leq k$:
\begin{itemize} 
	\item If $g$ is an addition gate, set $\displaystyle \vec{b}_{g,p} := \sum_{i \in [\ell]} \vec{b}_{g_i,p}$.
	\item If $g$ is a multiplication gate and $P_{g_1}$ has at most $d$ monomials. Then $\ell=2$. Compute $\vec{b}_{g,p}$ as
	\[
		\vec{b}_{g,p} = \sum_{z \in \{0,\ldots,p\}} \sum_{(\{x_{e_1},\ldots, x_{e_z}\}, \lambda) \in Q_{g_1}} \lambda\cdot   \mathsf{convolve}(\vec{b}_{g_2,p-z}, \{e_1,\ldots,e_z\}).
	\]
	\item Otherwise, by $d$-skewedness, $g$ is a multiplication gate and $P_{g_2}$ has at most $d$ monomials. Then compute $\vec{b}_{g,p}$ as
	\[
	\vec{b}_{g,p} = \sum_{z \in \{0,\ldots,p\}} \sum_{(\{x_{e_1},\ldots, x_{e_z}\}, \lambda) \in Q_{g_2}} \mathsf{convolve}(\vec{b}_{g_1,p-z}, \{e_1,\ldots,e_z\})\cdot\lambda.
	\]
\end{itemize}
It follows from \Ptwo{} of Theorem~\ref{thm:main-sr}, Lemma~\ref{lem:per} and Observations~\ref{obs:lin},~\ref{obs:scal} and~\ref{obs:init}, that $\vec{b}_{g,p}$ represents $\vec{a}_{g,p}$, for every $g,p$. Hence, by Observation~\ref{obs:query}, the solution to the instance of \textsc{Skewed Multilinear Monomial Summation over $\bS$} can be found as $(\vec{b}_{g,p}\cdot\vec{R})[\emptyset]$, where $s$ is the only sink of $C$ (e.g., the `output gate').
By combining \Pone{} of Theorem~\ref{thm:main-sr} with straightforward vector-matrix multiplication, this can be computed in $2^{k+O(\sqrt{k}\log^2k)}(d+n+|C|)\log n$ time, and hence Theorem~\ref{thm:mul} follows.

\section{Dynamic Representative Sets: Implementation}\label{sec:imp}
This section is devoted to the proof of Theorem~\ref{thm:main-sr}.
Throughout this section we shorthand $\preceq_\bS$ to $\preceq$. In the first three steps we assume $k$ is a square, and the fourth (less insightful but, due to the generality of theorem, non-trivial) step shows this assumption is without loss of generality.
	
\medskip\noindent\textbf{The preprocessing algorithm.}
Assuming $k$ is a square, we use the following parameters:
\[
s:=\sqrt{k} \in \mathbb{Z},\quad \text{ and }\quad  u:=k^2.
\]
Construct an $(n,k,u)$-splitter $\cP$ of size $k^{O(1)}\log n$ in time $k^{O(1)} n \log n$ with Lemma~\ref{lem:per}.
Construct a $(u,k,s)$-splitter $\cS$ of size  $k^{O(s)}$ in time $k^{O(s)}u^{O(1)}=k^{O(\sqrt{k})}$ with Lemma~\ref{lem:split}.
Construct a $(u,s)$-universal family $\cF$ of size $2^{s}s^{O(\log s)}\log u$ in time $2^{s}s^{O(\log s)}\log u=2^{O(\sqrt{k})}$ with Lemma~\ref{lem:unifam}.
By Definition~\ref{def:con}, we can compute $\pi(e)$ in $O(1)$ time for every $\pi \in \cP$ and $e \in [n]$.

\medskip\noindent\textbf{The matrices $\vec{L}$ and $\vec{R}$.}
Let $\cH = \cP \times \cS$, $h:=|\cH|$ and define $\vec{H}$ to be the following matrix whose rows are indexed by subsets of $[n]$ and whose columns are indexed by $\cH \times (2^{[u]})^{s}$:
\[
\vec{H}[A,((\pi,\sigma),A_1,\ldots,A_s)] := \left\llbracket \text{$\pi$ is injective on $A$ and for all }  i \in [s]: \sigma^{-1}(i)\cap \pi(A)=A_i \right\rrbracket.
\]
By expanding the matrix product, we obtain that $\left(\vec{H}\cdot\left(\vec{I}_{h} \otimes \vec{D}_{u,s}^{\otimes s}\right)\cdot\vec{H}^\intercal\right)[A,B]$ equals
\begin{equation}\label{eq:PDP}
	\begin{aligned}
	&\ \sum_{\substack{(\pi,\sigma)\in \cH\\ A_1,\ldots,A_s\subseteq [u] \\  B_1,\ldots,B_s \subseteq [u]}} \vec{H}[A,((\pi,\sigma),A_1,\ldots,A_s)]\cdot
	\left( \prod_{i=1}^s \vec{D}_{u,s}[A_i,B_i] \right)\cdot
	 \vec{H}^{\intercal}[((\pi,\sigma),B_1,\ldots,B_s),B] \\
	&= \sum_{(\pi,\sigma) \in \cH} \llbracket \text{$\pi$ is injective on $A$}  \rrbracket \cdot \llbracket \text{$\pi$ is injective on $B$} \rrbracket\cdot \prod_{i=1}^s \vec{D}_{u,s}[\sigma^{-1}(i) \cap \pi(A), \sigma^{-1}(i) \cap \pi(B)].
	\end{aligned}
\end{equation}
Observe that~\eqref{eq:PDP} equals $\vec{D}_{n,k}[A,B]$: If $|A|+|B| > k$ or $A \cap B \neq \emptyset$, by injectivity of $\pi$ we have that one of the entries of $\vec{D}_{u,s}$ is equal to $0$. Otherwise, since $\pi$ is an $(n,k,u)$-splitter and $\sigma$ is a $(u,k,s)$-splitter there exist $\pi$ such that $\pi$ is injective on $A \cup B$ and $\sigma$ such that for every $i=1,\ldots,s$ it holds that $|\sigma^{-1}(i) \cap (A \cup B)|\leq s$ and hence the corresponding term in~\eqref{eq:PDP} is a product of $\bar{1}$'s.

Our next step is to factorize $\vec{D}_{u,s}$. For every $A,B \subseteq [u]$, $F \in \cF$ and $p\in \{0,\ldots,s\}$ define
\[
	\vec{X}[A,(F,p)] :=	\llbracket A \subseteq F \wedge |A|\leq p\rrbracket \quad \text{and} \quad \vec{Y}[(F,p),B] := \llbracket F \cap B = \emptyset \wedge |B|\leq  s-p \rrbracket.
\]
We claim that $\vec{D}_{u,s}=\vec{X}\cdot\vec{Y}$: If $|A|+|B| > s$, there does not exist $p$ such that $|A| \leq p$ and $|B| \leq s-p$. If $A \cap B \neq \emptyset$, there does not exist $F$ such that $A\subseteq F \cap B=\emptyset$.
Otherwise, $A \cap B=\emptyset$ and $|A|+|B|\leq s$. Then, since $\cF$ is a $(u,s)$-universal family, there exists at least one $F \in \cF$ such that $A\subseteq F \cap B = \emptyset$ and due to idempotency of $\bS$ we have that $(\vec{X}\cdot\vec{Y})[A,B]=\bar{1}$.
Hence, if we define
\[
\vec{L} := \vec{H}\cdot\left(\vec{I}_{h} \otimes \vec{X}^{\otimes s}\right) \quad \text{and} \quad
\vec{R} := \left(\vec{I}_{h} \otimes \vec{Y}^{\otimes s}\right)\cdot\vec{H}^\intercal,
\]
then $\vec{D}_{n,k}=\vec{L}\cdot\vec{R}$. Note that $\vec{X}\in \{0,1\}^{2^{[u]}\times \ell}$ and $\vec{L} \in \{0,1\}^{2^{[n]} \times r}$, where $\ell=|\cF|\cdot(s+1)$ and
\[
\begin{aligned}
	r= h\cdot \ell^{s} &= h\cdot(|\cF|\cdot(s+1))^s\\
	&=|\cP|\cdot|\cS|\cdot\left(2^{s}k^{O(\log k)}\log k\cdot(s+1)\right)^s \\
	&= (k^{O(1)}\log n)k^{O(\sqrt{k})}k^{O(1)}2^kk^{O(\sqrt{k} \log k)}=2^{k+O(\sqrt{k}\log^2 k)}\log n.
\end{aligned}
\]
To see that \Pone{} holds, note that we can assume that the subset of $[n]$ is of cardinality at most $k$ since otherwise the asked entry equals $\bar{0}$. Then per Definition~\ref{def:con}, we can compute any entry of $\vec{L}$ in $k^{O(1)}$ time by directly evaluating the following expression:
\[
\begin{aligned}
	\vec{L}[A,((\pi,\sigma),(F_1,p_1),\ldots,(F_s,p_s))] =
	\begin{cases}
		\displaystyle\prod_{i=1}^s \vec{X}[\pi(A) \cap \sigma^{-1}(i),(F_i,p_i)] & \text{if $\pi$ is injective on $A$},\\
		0 & \text{otherwise},
	\end{cases}\\
	\vec{R}[((\pi,\sigma),(F_1,p_1),\ldots,(F_s,p_s)),B] =
	\begin{cases}
		\displaystyle\prod_{i=1}^s \vec{Y}[(F_i,p_i),\pi(B) \cap \sigma^{-1}(i)] & \text{if $\pi$ is injective on $B$},\\
		0 & \text{otherwise}.
	\end{cases}\\
\end{aligned}
\]

\noindent\textbf{Element convolution.}
We now prove \Ptwo{} of Theorem~\ref{thm:main-sr}. The algorithm~$\mathsf{convolve}$ is described in Algorithm~\ref{alg:elcons}.
A crucial ingredient is the subroutine $\mathsf{invert}$ that takes a vector $\vec{b^*} \in \bS^\ell$ (indexed by $\cF\times\{0,\ldots,s\}$) as input and outputs the unique minimal vector $\vec{a^*} \in \bS^{2^{[u]}}$ such that $\vec{b}\preceq\vec{a^*}\cdot\vec{L}$.
Algorithm $\mathsf{convolve}$ iterates over all restrictions $\vec{b^*}$ of $\vec{b}$ induced by the values picked at the loops at Lines~\ref{linfors1},~\ref{linfors2} and~\ref{linfors3} and replaces this restriction $\vec{b^*}$ with $\mathsf{invert}(\vec{X},\vec{b^*})\cdot\vec{C}_{u,\pi(e)}\cdot\vec{X}$.

\begin{algorithm}
	\caption{Element convolution algorithm $\mathsf{convolve}$ implementing \protect\Ptwo{} of Theorem~\ref{thm:main-sr}.}
	\label{alg:elcons}
	\begin{algorithmic}[1]
		\REQUIRE  $\mathsf{convolve}(\vec{b},e)$ 
		\FOR{$(\pi,\sigma) \in \cH$}\label{linfors1}
		\STATE Let $i := \sigma(\pi(e))$
		\FOR{$F_1,\ldots,F_{i-1},F_{i+1},\ldots,F_s \in \cF$}\label{linfors2}
		\FOR{$p_1,\ldots,p_{i-1},p_{i+1},\ldots,p_s\in\{0,\ldots,s\}$}\label{linfors3}
		\FOR{$F_i \in \cF$ and $p_i \in \{0,\ldots,s\}$ }
		\STATE Let $\vec{b^*}[(F_i,p_i)]:= \vec{b}[(\pi,\sigma,(F_1,p_1),\ldots,(F_s,p_s))]$
		\ENDFOR
		\STATE Let $\vec{a^*}:= \mathsf{invert}(\vec{X},\vec{b^*})$\label{ln:inbvert}
		\hfill\algcomment{described in Lemma~\ref{lem:sisr}}
		\FOR{$F_i \in \cF$ and $p_i \in \{0,\ldots,s\}$ }
		\STATE Let $\vec{b'}[(\pi,\sigma,(F_1,p_1),\ldots,(F_s,p_s))] := (\vec{a^*}\cdot\vec{C}_{u,\pi(e)}\cdot\vec{X})[(F_i,p_i)]$\label{lin:setbp}
		\ENDFOR
		\ENDFOR
		\ENDFOR
		\ENDFOR
		\STATE $\mathbf{return}\ \vec{b'}$
	\end{algorithmic}		
\end{algorithm}

The following lemma (stated in a setting of general binary matrices for convenience, but we mostly keep the notation in which we apply it) shows that over additively idempotent semirings one can find the unique minimal solution to a system of inequalities in polynomial time. While this is a known result (it is implied by e.g.~\cite{math12182904}), we give a self-contained proof for completeness.
\begin{lemma}\label{lem:sisr} 
	There is an algorithm $\mathsf{invert}(\vec{X},\vec{b^*})$ that takes as input a matrix $\vec{X} \in \{\bar{0},\bar{1}\}^{m \times \ell}$ and a vector $\vec{b^*} \in \bS^{\ell}$ and it outputs after $\poly(m)$ time and $\poly(m)$ additions, multiplications and $\mathrm{lcu}$ operations on elements in $\bS$ a vector $\vec{a^*} \in \bS^{m}$ such that $\vec{b^*} \preceq \vec{a^*}\cdot\vec{X}$ and $\vec{a^*}$ is minimal in the sense that for every $\vec{\hat{a}}$ satisfying $\vec{b^*} \preceq \vec{\hat{a}}\cdot\vec{X}$ we have $\vec{a^*} \preceq \vec{\hat{a}}$.
\end{lemma}
\begin{proof}
	For every $j \in [m]$, we define $\vec{a}^*$ as follows:
	\[
	\vec{a}^{*}[j]:= \mathrm{lcu}(\{ \vec{b^*}[i]: \vec{X}[i,j]=\bar{1}\}).
	\]
	Then $\vec{b^*}[i] \preceq_\bS \vec{a^*}[j]$ for each $i,j$ satisying $\vec{X}[i,j]=\bar{1}$, and hence $\vec{b}^*[i] \preceq \sum_{j:\vec{X}[i,j]=1}\vec{a}^*[j]=(\vec{a^*}\cdot\vec{X})[i]$.

	Suppose for a contradiction that $\vec{b}^* \preceq \vec{\hat{a}}\cdot\vec{X}$, but $\vec{a}^*\not\preceq \vec{\hat{a}}$. Then $\vec{b}^*\preceq (\vec{\hat{a}}+\vec{a}^*)\cdot\vec{X}$ and for some $j \in [m]$ it holds that $\vec{\hat{a}}[j]+\vec{a}^*[j] \prec \vec{a}^*[j]$. Therefore,  $\vec{a}^*[j] \neq \mathrm{lcu}(\{ \vec{b^*}[i]: \vec{X}[i,j]=\hat{1}\})$ since $\vec{\hat{a}}[j]+\vec{a}^*[j]$ also is an upper bound of all elements in $\{ \vec{b^*}[i]: \vec{X}[i,j]=1\}$, contradicting the choice of $\vec{a^*}$.
\end{proof}
For notational convenience, we define the following matrix for $v \in [u]$, $i \in [s]$ and $e \in [n]$:
\[
	\vec{J}_{v,i,e}[(\pi,\sigma),(\pi',\sigma')] := \left\llbracket (\pi,\sigma)=(\pi',\sigma'),\pi(e)=v, \text{ and } \sigma(v)=i \right\rrbracket.	
\]
Since $\vec{J}_{v,i,e}$ is equal to $\hat{0}$ on all off-diagonal entries, we have that, for any $e \in [n]$,  $\sum_{v \in [u], i \in [s]}\vec{J}_{v,i,e}=\vec{I}_h$, and also that $\vec{J}_{v,i,e}\cdot \vec{J}_{v',i',e}$ equals the all-$\bar{0}$ matrix if $v\neq v'$ or $i\neq i'$, and $\vec{J}_{v,i,e}\cdot \vec{J}_{v,i,e}=\vec{J}_{v,i,e}$.

We use that $\vec{C}_{n,e}$ and $\vec{H}$ commute in the following sense:
\begin{claim}\label{clm:comm} We have that $\vec{C}_{n,e}\cdot\vec{H} = \vec{H}\cdot\vec{\hat{C}}_e$, where we define	$\vec{\hat{C}}_e := \left( \sum_{v,i} \vec{J}_{v,i,e} \otimes \vec{I}^{\otimes i-1}_{2^u}\otimes  \vec{C}_{u,v} \otimes \vec{I}^{\otimes s-i}_{2^u}\right)$.
\end{claim}
\begin{proof}
	Note that the only non-zero entries of $\vec{\hat{C}}_e$ are
	\[
	\vec{\hat{C}}_e\Big[\Big((\pi,\sigma),A_1,\ldots,A_s\Big),\hspace{1em}\Big((\pi,\sigma),A_1,\ldots,A_{i-1},A_{i} \cup \pi(e),  A_{i+1},\ldots ,A_s\Big)\Big],
	\]
	whenever $\pi(e) \notin A_i$, where we shorthand $i=\sigma(\pi(e))$. We expand the matrix product and obtain
	\begin{align}\label{eq:condione}
	(\vec{C}_{n,e} \cdot \vec{H})[A,((\pi,\sigma),A_1,\ldots,A_s)]&=
	\left\llbracket
	\begin{aligned}\setcounter{equation}{\theequation+1}
		&\ \text{$e\notin A$, $\pi$ is injective on $A \cup e$, and}&\hypertarget{B1}{(\theequation a)}\\
		&\ \forall j \in [s] \setminus i: \sigma^{-1}(j)\cap \pi(A \cup e)=A_j \text{ and}&\hypertarget{B2}{(\theequation b)}\\
		&\  \sigma^{-1}(i)\cap \pi(A \cup e)=A_i&\hypertarget{B3}{(\theequation c)}
	\end{aligned}\setcounter{equation}{\theequation-1}
	\right\rrbracket.
	\intertext{Expanding the other matrix product, we also have that}\label{eq:conditwo}
	(\vec{H} \cdot \hat{\vec{C}}_e)[A,((\pi,\sigma),A_1,\ldots,A_s)]&=
	\left\llbracket
	\begin{aligned}\setcounter{equation}{\theequation+1}
		&\ \text{$\pi$ is injective on $A$, and}&\hypertarget{A1}{(\theequation a)}\\
		&\  \forall j \in [s] \setminus i: \sigma^{-1}(j)\cap \pi(A)=A_j \text{ and}&\hypertarget{A2}{(\theequation b)}\\
		&\ \pi(e) \in A_i, \sigma^{-1}(i)\cap \pi(A)=A_i \setminus \pi(e)\:&\hypertarget{A3}{(\theequation c)}
	\end{aligned}\setcounter{equation}{\theequation-1}
	\right\rrbracket.
	\end{align}
	Observe that \mylink{B2} is equivalent to \mylink{A2}, because $ \pi(e) \notin \sigma^{-1}(j)$ for $j\neq i$, since $i=\sigma(\pi(e))$.
	Note that \mylink{B1} implies \mylink{A1} and that \mylink{A3} implies \mylink{B3}.
	Moreover, \mylink{A3} implies that $\pi(e) \notin \pi(A)$ (and thus $e \notin A$) and hence $\mylink{A3}$ and $\mylink{A1}$ together imply $\mylink{B1}$.
	Finally, $\mylink{B3}$ and $\mylink{B1}$ together imply $\mylink{A3}$: $\pi(e) \in A_i$ since $\sigma^{-1}(i)\cap \pi(A \cup e)=A_i$ and $\sigma^{-1}(i)\cap \pi(A)=A_i \setminus \pi(e)$ since $\sigma^{-1}(i)\cap \pi(A \cup e)=A_i$ and $\pi$ is injective on $A \cup \{e\}$. Hence, the two conjunctions of conditions are equivalent and $\vec{C}_{n,e}\cdot\vec{H} = \vec{H}\cdot\vec{\hat{C}}_e$.
\end{proof}
Consider any iteration of the loops of Lines~\ref{linfors1},~\ref{linfors2} and~\ref{linfors3} corresponding to values $(\pi,\sigma)$, $F_1,\ldots,F_{i-1}$, $F_{i+1},\ldots,F_s$ and $p_1,\ldots,p_{i-1}$, $p_{i+1},\ldots,p_s$.
For every $A \subseteq [u]$ define
\[
	\vec{a}^*_{\mathsf{ext}}[(\pi,\sigma,(F_1,p_1),\ldots,(F_{i-1},p_{i-1}),A,(F_{i+1},p_{i+1}),\ldots,(F_s,p_s))]
\]
to be the value set to $\vec{a}^*[A]$ at Line~\ref{ln:inbvert} in this iteration.
Since every part of $\vec{b'}$ induced by $\pi,\sigma,F_1,\ldots,F_{i-1}$, $F_{i+1},\ldots,F_s$ and $p_1,\ldots,p_{i-1},p_{i+1},\ldots,p_s$ is defined as $\vec{a^*}\cdot\vec{C}_{u,v}\cdot\vec{X}$ we have that
\begin{equation}\label{eq:bprime}
\vec{b'} = \vec{a}^*_{\mathsf{ext}}\left(\sum_{v,i} \vec{J}_{v,i,e}\otimes \vec{I}_{\ell}^{\otimes i-1}\otimes\vec{C}_{u,v}\cdot \vec{X}\otimes \vec{I}_{\ell}^{\otimes s-i}\right).
\end{equation}
We now show that $\vec{b}\cdot\vec{R}\cdot \vec{C}^\intercal_{n,e}\preceq \vec{b'}\cdot\vec{R}$. By the properties of $\mathsf{invert}$ as stated in Lemma~\ref{lem:sisr}, we have that 
$\vec{b}^* \preceq \vec{a}^* \cdot \vec{X}$.
Hence we have the following global inequality:
\[
\vec{b} \preceq \vec{a}^*_{\mathsf{ext}}\left(\sum_{v,i} \vec{J}_{v,i,e}\otimes \vec{I}_{\ell}^{\otimes i-1}\otimes\vec{X}\otimes \vec{I}_{\ell}^{\otimes s-i}\right).
\]
Multiplying both sides with $\vec{R}\cdot \vec{C}^{\intercal}_{n,e}$ we see that $\vec{b}\cdot\vec{R}\cdot \vec{C}^\intercal_{n,e} \preceq$
\begin{align*}
	&\phantom{\preceq\ } \vec{a}^*_{\mathsf{ext}}\left(\sum_{v,i} \vec{J}_{v,i,e}\otimes \vec{I}_{\ell}^{\otimes i-1}\otimes\vec{X}\otimes \vec{I}_{\ell}^{\otimes s-i}\right)\vec{R}\cdot \vec{C}^{\intercal}_{n,e},\\
	\intertext{and using $\vec{H}^\intercal \cdot\vec{C}^\intercal_{n,e}=\vec{\hat{C}}^\intercal_e\cdot\vec{H}^\intercal$ (from Claim~\ref{clm:comm}) we can rewrite this into}
	&= \vec{a}^*_{\mathsf{ext}}\left(\sum_{v,i} \vec{J}_{v,i,e}\otimes \vec{I}_{\ell}^{\otimes i-1}\otimes\vec{X}\otimes \vec{I}_{\ell}^{\otimes s-i}\right)\cdot \left( \sum_{v,i} \vec{J}_{v,i,e} \otimes \vec{Y}^{\otimes s}\right)\vec{\hat{C}}^\intercal_e\cdot \vec{H}^\intercal\\
	\intertext{Using $\vec{\hat{C}}^\intercal_e=\sum_{v,i} \vec{J}_{v,i,e} \otimes \vec{I}^{\otimes i-1}_{2^u}\otimes  \vec{C}^\intercal_{u,v} \otimes \vec{I}^{\otimes s-i}_{2^u}$ and that all off-diagonal entries of $\vec{J}_{v,i,e}$ are $\bar{0}$ we can further rewrite into}
	&= \vec{a}^*_{\mathsf{ext}}\left(\sum_{v,i} \vec{J}_{v,i,e}\otimes \vec{I}_{\ell}^{\otimes i-1}\otimes\vec{D}_{u,s}\cdot\vec{C}^\intercal_{u,e}\otimes \vec{I}_{\ell}^{\otimes s-i}\right)\cdot \left( \sum_{v,i} \vec{J}_{v,i,e} \otimes \vec{Y}^{\otimes i-1}\otimes \vec{I}_{2^{u}} \otimes \vec{Y}^{\otimes s-i}\right) \vec{H}^\intercal\\
	\intertext{applying Lemma~\ref{lem:comm} we obtain that}
	&= \vec{a}^*_{\mathsf{ext}}\left(\sum_{v,i} \vec{J}_{v,i,e}\otimes \vec{I}_{\ell}^{\otimes i-1}\otimes\vec{C}_{u,e}\cdot\vec{X}\otimes \vec{I}_{\ell}^{\otimes s-i}\right)\cdot \left( \sum_{v,i} \vec{J}_{v,i,e} \otimes \vec{Y}^{\otimes s}\right) \vec{H}^\intercal.\\
	\intertext{and we can use~\eqref{eq:bprime} and that $\sum_{v,i}\vec{J}_{v,i,e}=\vec{I}_h$ to bound as}
	&\preceq \vec{b'}\left(\vec{I}_h \otimes \vec{Y}^{\otimes s}\right)\vec{H}^\intercal\\
	&=\vec{b'}\cdot\vec{R}.
\end{align*}
It remains to show that $\vec{b} \preceq \vec{a} \cdot \vec{L}$ implies $\vec{b'} \preceq \vec{a}\cdot \vec{C}_{n,e} \cdot \vec{L}$.
Define
\begin{align*}
	\vec{\hat{a}}_{\mathsf{ext}} &:= \vec{a} \cdot \vec{H} \left(\sum_{v,i} \vec{J}_{v,i,e}\otimes \vec{X}^{\otimes i-1}\otimes\vec{I}_{2^u}\otimes \vec{X}^{\otimes s-i}\right),\\
	\intertext{and also its more local version:}
	\vec{\hat{a}}[A] &:= \vec{\hat{a}}_{\mathsf{ext}}[(\pi,\sigma,(F_1,p_1),\ldots,(F_{i-1},p_{i-1}),A,(F_{i+1},p_{i+1}),\ldots,(F_s,p_s))].\\
	\intertext{Observe that}
	&(\vec{\hat{a}} \cdot\vec{X})[(F_i,p_i)]=(\vec{a}\cdot\vec{L})[(\pi,\sigma,(F_1,p_1),\ldots,(F_s,p_s))]
\end{align*}
Hence, by the assumption $\vec{b}\preceq \vec{a}\cdot\vec{L}$ we have $\vec{b}^* \preceq \vec{\hat{a}}\cdot\vec{X}$.
In each iteration we thus have by Lemma~\ref{lem:sisr} that 
$\vec{a}^* \preceq \vec{\hat{a}}$ and hence $\vec{a}^*_{\mathsf{ext}} \preceq \vec{\hat{a}}_{\mathsf{ext}}$.
Multiplying both sides of $\vec{a}^*_{\mathsf{ext}} \preceq \vec{\hat{a}}_{\mathsf{ext}}$ with $\sum_{v,i} \vec{J}_{v,i,e}\otimes \vec{I}_{\ell}^{\otimes i-1}\otimes\vec{C}_{u,v}\cdot \vec{X}\otimes \vec{I}_{\ell}^{\otimes s-i}$ we obtain by~\eqref{eq:bprime} that
\begin{align*}
	\vec{b'} &\preceq
	\vec{\hat{a}}_\mathsf{ext}\ \ \left(\sum_{v,i} \vec{J}_{v,i,e}\otimes \vec{I}_{\ell}^{\otimes i-1}\otimes\vec{C}_{u,v}\cdot \vec{X}\otimes \vec{I}_{\ell}^{\otimes s-i}\right)\\
	&= \vec{a} \cdot \vec{H} \left(\sum_{v,i} \vec{J}_{v,i,e}\otimes \vec{X}^{\otimes i-1}\otimes\vec{I}_{2^u}\otimes \vec{X}^{\otimes s-i}\right)\cdot s\left(\sum_{v,i} \vec{J}_{v,i,e}\otimes \vec{I}_{\ell}^{\otimes i-1}\otimes\vec{C}_{u,v}\cdot \vec{X}\otimes \vec{I}_{\ell}^{\otimes s-i}\right)\\
	\intertext{Using that all off-diagonal entries of $\vec{J}_{v,i,e}$ are $\bar{0}$, we can rewrite this into}
	&= \vec{a} \cdot \vec{H} \left(\sum_{v,i} \vec{J}_{v,i,e}\otimes \vec{I}_{2^u}^{\otimes i-1}\otimes\vec{C}_{u,v}\otimes \vec{I}_{2^u}^{\otimes s-i}\right)\cdot \left(\sum_{v,i} \vec{J}_{v,i,e}\otimes \vec{X}^{\otimes s}\right)\\
	&= \vec{a}\cdot\vec{H}\cdot\vec{\hat{C}}_e\cdot \left(\sum_{v,i} \vec{J}_{v,i,e}\otimes \vec{X}^{\otimes s}\right)\\
	\intertext{Using Claim~\ref{clm:comm} and that $\sum_{v,i}\vec{J}_{v,i,e}=\vec{I}_h$, this is equal to}
	&=  \vec{a} \cdot \vec{C}_{n,e}\cdot \vec{L}.
\end{align*}

\subsection*{Step 4: Justification of the assumption that $k$ is a square}
We now prove Theorem~\ref{thm:main-sr}, assuming it holds whenever $k$ is a square number.
Define $s:=\left\lceil\sqrt{k}\right\rceil$, $d:=s^2-k$, $n':=n+d$ and $k':=s^2$. Apply the theorem with $n'$ and $k'$ to obtain a factorization $\vec{D}_{n',k'}=\vec{L'}\cdot\vec{R'}$ of rank $r=2^{k'+O(\sqrt{k'}\cdot \log^2k')}\log n'=2^{k+O(\sqrt{k}\cdot \log^2k)}\log n$.
Define a matrix $\vec{P} \in \bS^{2^{n}\times 2^{n'}}$ as follows: For every $A \subseteq [n]$ and $A'\subseteq [n']$ we set
\[
\vec{P}[A,A'] = \llbracket A \cup \{n+1,\ldots,n+d\} = A' \rrbracket.
\]
Observe that $\vec{P}\cdot\vec{C}_{n,e}=\vec{C}_{n,e}\cdot\vec{P}$, if $e \in [n]$.
Define $\vec{L}$ and $\vec{R}$ as follows:
$\vec{L}:=\vec{P}\cdot\vec{L'}$ and $\vec{R}$ is equal to $\vec{R'}$ restricted to subsets of $\{1,\ldots,n\}$. Note that, for any $A,B\subseteq [n]$ we have that $(\vec{L}\cdot\vec{R})[A,B]=\vec{D'}_{n',k'}[A \cup \{n+1,\ldots,n+d\},B]$.
Hence, $\vec{L}\cdot\vec{R}=\vec{D}_{n,k}$.
Moreover, any entry of $\vec{L}$ and $\vec{R}$ can be computed in $k^{O(1)}$ time since this also holds for $\vec{L'}$ and $\vec{R'}$. Hence \Pone{} follows.

We now show \Ptwo. Let $\mathsf{convolve}$ be the algorithm that is given by Theorem~\ref{thm:main-sr} with $n',k'$.
Let $e \in [n]$ and define $\vec{b'}:=\mathsf{convolve}(\vec{b},e)$.
We have that $\vec{b}\cdot\vec{R'}\cdot\vec{C}^\intercal_{n',e} \preceq_\bS \vec{b'}\cdot\vec{R'}$ and hence $\vec{b}\cdot\vec{R}\cdot\vec{C}^\intercal_{n,e} \preceq_\bS \vec{b'}\cdot\vec{R}$.
Also, if $\vec{a} \in \bS^{2^{[n]}}$ such that $\vec{b} \preceq_\mathbb{S} \vec{a}\cdot\vec{L} = (\vec{a} \cdot\vec{P}) \cdot \vec{L'}$ then it holds that $\vec{b'} \preceq_\mathbb{S} (\vec{a}\cdot\vec{P})\cdot \vec{C}_{n,e} \cdot \vec{L'}=\vec{a}\cdot\vec{C}_{n,e}\cdot\vec{L}$, where the inequality follows from the mentioned observation $\vec{P}\cdot\vec{C}_{n,e}=\vec{C}_{n,e}\cdot\vec{P}$.

\paragraph{Acknowledgements.}
The author thanks Tomohiro Koana for extensive comments on this manuscript and helpful discussions and Carla Groenland for comments on an earlier version of this manuscript.

\small
\bibliographystyle{alpha}
\bibliography{refs}

\end{document}